\documentclass[11pt,onecolumn]{article}
\usepackage{amsmath,amsthm,amssymb,amsfonts,epsfig,graphicx}
\usepackage[linesnumbered,ruled,vlined]{algorithm2e}
\usepackage{color}
\usepackage{dsfont}
\usepackage{relsize}

\usepackage[left=1.0in,top=1.0in,right=1.05in,bottom=0.5in,nohead,textheight=10in,footskip=0.3in]{geometry}

\newtheorem{theorem}{Theorem}
\newtheorem{corollary}[theorem]{Corollary}
\newtheorem{conjecture}{Conjecture}
\newtheorem{lemma}[theorem]{Lemma}
\newtheorem{definition}[theorem]{Definition}

\newtheorem{remark}{Remark}

\DeclareMathOperator{\Exp}{Exp}

\newcommand{\Pol}[2][]{    \ifthenelse{\equal{#1}{}}{  \ensuremath{\operatorname{Pol}(#2)}    }{    \ensuremath{\operatorname{Pol}^{(#1)}(#2)}    }          }
\newcommand{\fPolplus}[2][]{    \ifthenelse{\equal{#1}{}}{  \ensuremath{\operatorname{supp}(#2)}    }{    \ensuremath{\operatorname{supp}^{(#1)}(#2)}    }          }
\newcommand{\fPol}[2][]{    \ifthenelse{\equal{#1}{}}{  \ensuremath{\operatorname{fPol}(#2)}    }{    \ensuremath{\operatorname{fPol}^{(#1)}(#2)}    }          }
\newcommand{\Ocore}[1]{\ensuremath{\calO^{\tt core}_{#1}}}
\newcommand{\Bcore}[1]{\ensuremath{{\cal B}^{\tt core}_{#1}}}

\newcommand{\BcorePi}{\ensuremath{{\Pi}^{\perp}}}

\newcommand{\OcorePi}{\ensuremath{{\cal O}[\Pi^{\perp}]}}

\DeclareMathOperator*{\argmin}{arg\,min}

\newcommand{\eqdef}{{\stackrel{\mbox{\tiny \tt ~def~}}{=}}}

\newcommand{\vcsp}[1]{\ensuremath{\operatorname{VCSP}(#1)}}
\newcommand{\VCSP}[1]{\ensuremath{\operatorname{VCSP}(#1)}}

\newcommand{\Feas}[1]{\ensuremath{\operatorname{Feas}(#1)}}
\newcommand{\Feasplus}[1]{\ensuremath{\operatorname{Feas}(#1)}}

\newcommand{\TFeas}[1]   {\ensuremath{T_{#1}}}
\newcommand{\TLPPROBE}[1]{\ensuremath{T_{#1}^\ast}}

\DeclareMathOperator{\BLP}{BLP}
\DeclareMathOperator{\dom}{\textup{\texttt{dom}}}
\DeclareMathOperator{\supp}{supp}
\DeclareMathOperator{\size}{\textup{\texttt{size}}}
\DeclareMathOperator{\coresize}{\textup{\texttt{core-size}}}

\newcommand{\bx}{\mbox{\boldmath $x$}}
\newcommand{\bxS}{\mbox{\scriptsize\boldmath $x$}}
\newcommand{\myparagraph}[1]{\noindent{\textbf{#1}}\quad}

\newcommand{\HP}[1]{\ensuremath{\langle #1\rangle}}

\newcommand{\Qc}{\mbox{$\overline{\mathbb Q}$}}

\def\calB{{\cal B}}
\def\calC{{\cal C}}

\def\calF{{\cal F}}
\def\calG{{\cal G}}
\def\calI{{\cal I}}
\def\calJ{{\cal J}}

\def\calL{{\cal L}}

\def\calO{{\cal O}}

\def\calOplus{{\widehat\calO}}

\begin{document}

\def\myparagraph#1{\vspace{2pt}\noindent{\bf #1~~}}




\title{\Large\bf  \vspace{0pt} Testing the complexity of a valued CSP language}
\author{Vladimir Kolmogorov \\ \normalsize Institute of Science and Technology Austria \\ {\normalsize\tt vnk@ist.ac.at}}
\date{}
\maketitle

\begin{abstract}
A {\em Valued Constraint Satisfaction Problem} (VCSP)
provides a common framework that can express a wide range of discrete optimization
problems. A  VCSP instance is given by a finite set of variables, a finite domain of labels, and an
objective function to be minimized.
This function is represented as a sum of terms where each term depends
on a subset of the variables. To obtain different classes of optimization
problems, one can restrict all terms to come from a fixed set $\Gamma$ of cost functions,
called a {\em language}. 

Recent breakthrough results have established a complete complexity
classification of such classes with respect to language $\Gamma$:
if all cost functions in $\Gamma$ satisfy a certain algebraic condition
then all $\Gamma$-instances can be solved in polynomial time, otherwise the problem is NP-hard.
Unfortunately, testing this condition for a given language $\Gamma$ is known to be NP-hard.
We thus study exponential algorithms for this {\em meta-problem}.
We show that the tractability condition of a finite-valued language $\Gamma$
can be tested in $O(\sqrt[3]{3}^{\,|D|}\cdot poly(\size(\Gamma)))$ time,
where $D$ is the domain of $\Gamma$ and $poly(\cdot)$ is some fixed polynomial.
We also obtain a matching lower bound under the {\em Strong Exponential Time Hypothesis} (SETH).
More precisely, we prove that for any constant $\delta<1$ there is no $O(\sqrt[3]{3}^{\,\delta|D|})$ 
algorithm, assuming that SETH holds.
\end{abstract}

\section{Introduction}
Minimizing functions of discrete variables represented as a sum of low-order terms 
is a  ubiquitous problem occurring in many real-world applications.
Understanding complexity of different classes of such optimization problems is
thus an important task. In a prominent {\em VCSP framework}
(which stands for {\em Valued Constraint Satisfaction Problem}) a class
is parameterized by a set $\Gamma$ of cost functions of the form $f:D^n\rightarrow\mathbb Q\cup\{\infty\}$ that are allowed to appear
as terms in the objective. Set $\Gamma$ is usually called a {\em language}.

Different types of languages give rise to many interesting classes.
A widely studied type
is {\em crisp} languages $\Gamma$,
in which all functions $f$ are $\{0,\infty\}$-valued.
They correspond to {\em Constraint Satisfaction Problems (CSPs)},
whose goal is to decide whether a given instance has a feasible solution.
Feder and Vardi conjectured in~\cite{feder98:monotone} that there exists a dichotomy
for CSPs, i.e.\ every crisp language $\Gamma$
is either tractable or NP-hard. This conjecture was refined by Bulatov, Krokhin and Jeavons~\cite{bulatov05:classifying},
who proposed a specific algebraic condition that should separate tractable languages from NP-hard ones.
The conjecture was verified for many special cases~\cite{Schaefer78:complexity,bulatov06:3-elementjacm,barto09:siam,barto11:lics,Bulatov11:conservative},
and was finally proved in full generality by Bulatov~\cite{Bulatov:FOCS17} and  Zhuk~\cite{Zhuk:FOCS17}.

At the opposite end of the VCSP spectrum are the finite-valued CSPs,
in which functions do not take infinite values. In such VCSPs, the feasibility aspect
is trivial, and one has to deal only with the optimization issue. One polynomial-time
algorithm that solves tractable finite-valued CSPs is based on the so-called basic linear
programming (BLP) relaxation, and its applicability (also for the general-valued case)
was fully characterized by Kolmogorov, Thapper and \v{Z}ivn\'y~\cite{kolmogorov15:power}.
The complexity of finite-valued CSPs was completely classified by Thapper and \v{Z}ivn\'y~\cite{tz16:jacm}, where it is shown that all finite-valued CSPs
not solvable by BLP are NP-hard.

A dichotomy is also known to hold for general-valued CSPs, i.e.\ when cost functions in $\Gamma$ are allowed
to take arbitrary values in $\mathbb Q\cup\{\infty\}$.
First, Kozik and Ochremiak showed~\cite{Kozik15:algebraic} that languages that do not satisfy
a certain algebraic condition are NP-hard.
Kolmogorov, Krokhin and Rol\'{i}nek then proved~\cite{KKZ:SICOMP17}
that all other languages are tractable, assuming the (now established) dichotomy for crisp languages conjectured in~\cite{bulatov05:classifying}.

In this paper languages $\Gamma$ that satisfy the condition in~\cite{Kozik15:algebraic} are called {\em solvable}.
Since optimization problems encountered in practice often come without any guarantees,
it is natural to ask what is  the complexity of checking solvability of a given language $\Gamma$.
We envisage that an efficient algorithm for this problem could help in theoretical investigations,
and  could also facilitate designing
optimization approaches for tackling specific tasks.

Checking solvability of a given language is known as a {\em meta-problem} or a {\em meta-question} in the literature.
Note that it can be solved in polynomial time for languages on a fixed domain $D$
(since the solvability condition can be expressed by a linear program
with $O(|D|^{|D|^{m}})$ variables and polynomial number of constraints,
where $m=2$ if the language is finite-valued and $m=4$ otherwise).
This naive solution, however, becomes very inefficient if $D$ is a part of the input
(which is what we assume in this paper).

The meta-problem above  was studied by Thapper and \v{Z}ivn\'y for finite-valued languages~\cite{tz16:jacm},
and by Chen and Larose for crisp languages~\cite{ChenLarose:17}.
In both cases it was shown to be NP-complete.
We therefore focus on exponential-time algorithms. We obtain the following results for
the problem of checking solvability of a given finite-valued language $\Gamma$:
\begin{itemize}
\item An algorithm with complexity $O(\sqrt[3]{3}^{\,|D|}\cdot poly(\size(\Gamma)))$,
where $D$ is the domain of $\Gamma$ and $poly(\cdot)$ is some fixed polynomial.
\item Assuming the {\em Strong Exponential Time Hypothesis} (SETH),
we prove that for any constant $\delta<1$ the problem cannot be solved
in $O(\sqrt[3]{3}^{\,\delta|D|}\cdot poly(\size(\Gamma)))$ time.
\end{itemize}
We also present a few weaker results for general-valued languages (see Section~\ref{sec:results}).

\paragraph{Other related work} There is a vast literature devoted to exponential-time algorithms
for various problems, both on the algorithmic side and on the hardness side.
Hardness results usually assume one of the following two hypotheses~\cite{SETH0,SETH1,SETH2}.
\begin{conjecture}[Exponential Time Hypothesis (ETH)]\label{conj:ETH}
Deciding satisfiability of a $3$-CNF-SAT formula
on $n$ variables cannot be solved in $O(2^{o(n)})$ time.
\end{conjecture}

\begin{conjecture}[Strong Exponential Time Hypothesis (SETH)]\label{conj:SETH}
For any $\delta<1$ there exists integer~$k$ such
that deciding satisfiability of a $k$-CNF-SAT formula
on $n$ variables cannot be solved in $O(2^{\delta n})$ time.
\end{conjecture}

Below we discuss some results specific to CSPs. Let $(k,d)$-CSP be the class of CSP problems
on a $d$-element domain where each constraint involves at most $k$ variables. The number of variables
in an instance will be denoted as $n$. A trivial exhaustive search for a $(k,d)$-CSP instance runs in $O^\ast(d^n)$ time,
where notation $O^\ast(\cdot)$ hides factors polynomial in the size of the input.
For $(2,d)$-CSP instances the complexity can be improved to $O^\ast((d-1)^n)$~\cite{razgon:05}.
Some important subclasses of $(2,d)$-CSP can even be solved in $O^\ast(2^{O(n)})$ time.
For example,~\cite{lawler:76} and~\cite{bjoerklund:09} developed respectively $O(2.45^n)$ and $O^\ast(2^n)$ algorithms
for solving the $d$-coloring problem.
On the negative side, ETH is known to have the following implications:
\begin{itemize}
\item The $(2,d)$-CSP problem cannot be solved in $d^{o(n)}=2^{o(n\log d)}$ time~\cite{traxler:08}.
\item The {\sc Graph Homomorphism} problem cannot be solved in $2^{o\left(\frac{n\log d}{\log\log d}\right)}$ time~\cite{Fomin:ICALP15}.
(This problem can be viewed as a special case of $(2,d)$-CSP, in which a single binary relation
is applied to different pairs of variables).
\end{itemize}
Recently, exponential-time algorithms for crisp NP-hard languages have been studied using algebraic techniques~\cite{jonsson:jcss17,jonsson:mfcs17,Lagerkvist:arXiv18}.
For example,~\cite{jonsson:mfcs17} showed that the following conditions are equivalent, assuming the (now proved) algebraic CSP dichotomy conjecture:
(a) ETH fails; 
(b) there exists a finite crisp NP-hard language $\Gamma$ that can be solved in subexponential time
(i.e.\ all $\Gamma$-instances on $n$ variables can be solved in $O(2^{o(n)})$ time);
(c) all finite crisp NP-hard languages $\Gamma$ can be solved in subexponential time.

The rest of the paper is organized as follows: Section \ref{sec:background} gives a background on the VCSP framework,
and Section~\ref{sec:results} presents our results. All proofs are given in Section~\ref{sec:proofs} and Appendices~$\mbox{\ref{sec:core:proofs}-\ref{sec:core:gstar:proof}}$.

\section{Background}\label{sec:background}
We denote $\Qc=\mathbb Q\cup\{\infty\}$, where $\infty$ is the positive infinity.
A function of the form $f:D^n\rightarrow \Qc$ will be called a {\em cost function over $D$ of arity $n$}.
We will always assume that the set $D$ is finite.
The {\em effective domain} of $f$ is the set $\dom f=\{x\mid f(x)<\infty\}$.
Note that $\dom f$ can be viewed both as an $n$-ary relation over $D$
and as a function $D^n\rightarrow\{0,\infty\}$. 
We assume that $f$
is represented as a list of pairs $\{(x,f(x))\::\:x\in \dom f\}$.
Accordingly, 
we define $\size(f)=\sum\nolimits_{x\in\dom f}[n\log|D|+\size(f(x))]$,
where the size of a rational number $p/q$ (for integers $p,q$) is $\log(|p|+1)+\log |q|$.

\begin{definition}
A valued constraint satisfaction language $\Gamma$ over domain $D$ is a set of cost functions 
 $f:D^n\rightarrow \Qc$, where the arity $n$ depends on $f$ and may
be different for different functions in $\Gamma$. The domain of $\Gamma$ will be denoted as $D_\Gamma$.
For a finite $\Gamma$ we define $\size(\Gamma)=|D|+\sum_{f\in\Gamma}\size(f)$.
\end{definition}

A language $\Gamma$ is called {\em finite-valued} if all functions $f\in\Gamma$ take finite (rational) values.
It is called {\em crisp} if all functions $f\in\Gamma$ take only values in $\{0,\infty\}$.
We denote $\Feas\Gamma=\{\dom f\:|\:f\in\Gamma\}$ to be the crisp language obtained from $\Gamma$ in a natural way.
Throughout the paper, for a subset $A\subseteq D$ we use $u_A$ to denote the unary function $D\rightarrow \{0,\infty\}$
with $\argmin u_A=A$. (Domain $D$ should always be clear from the context). For a label $a\in D$ we also write $u_a=u_{\{a\}}$ for brevity.

\begin{definition}\label{def:vcsp}
An instance $\calI$ of the valued constraint
satisfaction problem (VCSP) is a function $D^V\rightarrow \Qc$ given
by
\begin{equation}
f_\calI(x)\ =\ \sum_{t\in T} f_t(x_{v(t,1)},\ldots,x_{v(t,n_t)})      \label{eq:VCSPinst}
\end{equation}
It is specified by a finite set of variables $V$, finite set of terms
$T$, cost functions $f_t : D^{n_t}\rightarrow\Qc$ of arity $n_t$ and
indices $v(t, k)\in V$ for $t\in T , k =1,\ldots, n_t$. 
A solution to $\calI$ is a labeling $x\in D^V$ with the
minimum total value.
The size of $\calI$ is defined as $\size(\calI)=|V|+|D|+\sum_{t\in T}\size(f_t)$.

The instance $\calI$ is called a $\Gamma$-instance if all terms $f_t$ belong to $\Gamma$.
\end{definition}
The set of all $\Gamma$-instances will be denoted as $\vcsp{\Gamma}$.
A finite language $\Gamma$ is called {\em tractable} if all instances $\calI\in\vcsp{\Gamma}$
can be solved in polynomial time, and it is {\em NP-hard} if the corresponding
optimization problem is NP-hard.
A long sequence of works culminating with recent breakthrough papers~\cite{Bulatov:FOCS17,Zhuk:FOCS17}
has established that every finite language $\Gamma$ is either tractable or NP-hard.

\subsection{Polymorphisms and cores}

Let $\calO_D^{(m)}$ denote the set of all operations $g:D^m\rightarrow D$ and let
 $\calO_D=\bigcup_{m\ge 1}{\calO_D^{(m)}}$. When $D$ is clear from the context, we will sometimes write simply $\calO^{(m)}$ and $\calO$.

Any language $\Gamma$ defined on $D$ can be associated
with a set of operations on $D$,
known as the polymorphisms of $\Gamma$, which allow one to combine (often in a useful way)
several feasible assignments into a new one.
\begin{definition}
\label{def:polymorphism}
An operation $g\in \calO_D^{(m)}$ is a \emph{polymorphism} of a cost function $f:D^n\rightarrow \Qc$ if,
for any $x^1,x^2,\ldots,x^m \in \dom f$,
we have that $g(x^1,x^2,\ldots,x^m)\in \dom f$ where $g$ is applied component-wise.

For any valued constraint language $\Gamma$ over a set $D$,
we denote by 
$\Pol[m]\Gamma$ 
the set of all operations on $\calO^{(m)}_D$ which are polymorphisms of every
$f \in \Gamma$.
We also let $\Pol\Gamma=\bigcup_{m\ge 1}\Pol[m]\Gamma$.
\end{definition}

Clearly, if $g$ is a polymorphism of a cost function $f$, then $g$ is also a polymorphism of $\dom f$.
For $\{0,\infty\}$-valued functions, which naturally correspond to relations, the notion
of a polymorphism defined above coincides with the standard notion of a polymorphism for relations.
Note that the projections (aka dictators), i.e. operations of the form $e_n^i(x_1,\ldots,x_n)=x_i$,  are polymorphisms of all valued constraint languages.
Polymorphisms play the key role in the algebraic approach to the CSP, but, for VCSPs, more general constructs are necessary, which we now define.

\begin{definition}
An $m$-ary \emph{fractional operation} $\omega$ on $D$ is a probability distribution on $\calO_D^{(m)}$.
The support of $\omega$ is defined as $\supp(\omega)=\{g\in \calO_D^{(m)}\mid \omega(g)>0\}$.
\end{definition}
For an operation $g\in\calO^{(m)}$ we will denote $\chi_g$ to be characteristic vector of $g$,
i.e.\ the fractional operation with $\chi_g(g)=1$ and $\chi_g(h)=0$ for $h\ne g$.

\begin{definition}\label{def:fpol}
A $m$-ary fractional operation $\omega$ on $D$ is said to be a \emph{fractional polymorphism} of a
cost function $f:D^n\rightarrow \Qc$ if, for any $x^1,x^2,\ldots,x^m \in \dom f$,
we have
\begin{equation}
\sum_{g\in\supp(\omega)}{\omega(g)f(g(x^1,\ldots,x^m))} \le \frac{1}{m}(f(x^1)+\ldots+f(x^m)).
\label{eq:wpol-dist}
\end{equation}

For a constraint language $\Gamma$, $\fPol[m]\Gamma$ will denote the set of all $m$-ary fractional operations that are fractional polymorphisms
of each function in $\Gamma$. Also, let 
$\fPol\Gamma=\bigcup_{m\ge 1}\fPol[m]\Gamma$,
$\fPolplus[m]\Gamma=\bigcup_{\omega\in\fPol[m]\Gamma} \fPolplus\omega$
and
$\fPolplus\Gamma=\bigcup_{m\ge 1}\fPolplus[m]\Gamma$.


(It is easy to check that $\fPolplus\Gamma\subseteq\Pol\Gamma$, and $\fPolplus\Gamma=\Pol\Gamma$ if $\Gamma$ is crisp).
\end{definition}

Next, we will need the notion of {\em cores}. 
\begin{definition}
Language $\Gamma$ on domain $D$ is called a {\em core} if all operations $g\in\fPolplus[1]\Gamma$ are bijections.
Subset $B\subseteq D$ is called a {\em core of $\Gamma$} if $B=g(D)$ for some operation $g\in\fPolplus[1]\Gamma$
and the language $\Gamma[B]$ is a core, where $\Gamma[B]$
is the language on domain $B$ obtained by restricting each function $f:D^n\rightarrow\Qc$
in $\Gamma$ to $B^n$.
\end{definition}
The following facts are folklore knowledge. We do not know an explicit reference (at least in the case of general-valued languages), so we prove them in Appendix~\ref{sec:core:proofs} for completeness.
\begin{lemma}\label{lemma:core-basics}
Let $B$ be a subset of $D=D_\Gamma$ such that $B=g(D)$ for some $g\in\fPolplus[1]\Gamma$. \\
(a) Set $B$ is a core of $\Gamma$ if and only if  $|B|=\coresize(\Gamma)\eqdef \min\;\{ |g(D)| \::\: g\in\fPolplus[1]\Gamma \}$. \\
(b) There exists vector $\omega\in\fPol[1]\Gamma$ such that $g(D)\subseteq B$ for all $g\in\supp(\omega)$. 
Furthermore, if $B$ is a core of $\Gamma$ then such $\omega$ can be chosen so that $g(a)=a$ for all $g\in\supp(\omega)$ and $a\in B$. \\
(c) Let $\calI$ be a $\Gamma$-instance on variables $V$. Then $\min_{x\in B^V}f_\calI(x)=\min_{x\in D^V}f_\calI(x)$.
\end{lemma}
For a language $\Gamma$ we denote $\Bcore\Gamma$ to be the set of subsets $B\subseteq D$ which are cores of $\Gamma$,
and $\Ocore\Gamma$ to be set of operations $g\in\fPolplus[1]\Gamma$ such that $g(D)\in \Bcore\Gamma$
(or equivalently such that $|g(D)|=\coresize(\Gamma)$). 

\subsection{Dichotomy theorem}

Several types of operations play a special role in the algebraic approach to (V)CSP.

\begin{definition}
An operation $g\in \calO_D^{(m)}$ is called
\begin{itemize}
\item \emph{idempotent} if $g(x,\ldots,x)=x$ for all $x\in D$;
\item \emph{cyclic} if $m\ge 2$ and $g(x_1,x_2,\ldots,x_m)=g(x_2,\ldots,x_m,x_1)$ for all $x_1,\ldots,x_m\in D$;
\item \emph{symmetric} if $m\ge 2$ and $g(x_1,x_2,\ldots,x_m)=g(x_{\pi(1)},x_{\pi(2)},\ldots,x_{\pi(m)})$ for all $x_1,\ldots,x_m\in D$, and any permutation $\pi$ on $[m]$;
\item \emph{Siggers} if $m=4$ and $g(r,a,r,e)=g(a,r,e,a)$ for all $a,e,r\in D$.
\end{itemize}
A fractional operation $\omega$ is said to be idempotent/cyclic/symmetric if all operations in $\supp(\omega)$ have the corresponding property.
\end{definition}
Note,  the Siggers operation is traditionally defined in the literature as an {\bf idempotent} operation $g$
satisfying $g(r,a,r,e)=g(a,r,e,a)$. Here we follow the terminology in~\cite{CSPsurvey:17} that does not require idempotency.
(In~\cite{ChenLarose:17} such an operation was called {\em quasi-Siggers}). 

We can now formulate the dichotomy theorem.
\begin{theorem}\label{th:dichotomy}
Let $\Gamma$ be a language.
If the core of $\Gamma$ admits a cyclic fractional polymorphism then $\Gamma$ is tractable~\cite{KKZ:SICOMP17,Bulatov:FOCS17,Zhuk:FOCS17}.
Otherwise $\Gamma$ is NP-hard~\cite{Kozik15:algebraic}.
\end{theorem}

We will call languages $\Gamma$ satisfying the condition of Theorem~\ref{th:dichotomy} {\em solvable}.
The following equivalent characterizations of solvability are either known or 
can be easily be derived from previous work~\cite{Siggers:1,Siggers:2,kolmogorov15:power,Kozik15:algebraic}
 (see Appendix~\ref{sec:lemma:solvability:equiv:proof}):
\begin{lemma}\label{lemma:solvability:equiv}
Let $\Gamma$ be a language and $g\in\fPolplus[1]\Gamma$. The following conditions are equivalent:
\begin{enumerate}
\item[(a)] $\Gamma$ is solvable.
\item[(b)] $\Gamma$ admits a cyclic fractional polymorphism of some arity $m\ge 2$.
\item[(c)] $\fPolplus\Gamma$ contains a Siggers operation.
\item[(d)] $\Gamma\cup\{u_a\:|\:a\in B\}$ is solvable for any core $B$ of $\Gamma$. 
\item[(e)] $\Gamma[g(D_\Gamma)]$ is solvable.
\end{enumerate}
Furthermore, a finite-valued language $\Gamma$  is solvable if and only if it admits a symmetric fractional polymorphism of arity $2$.
\end{lemma}
Note that checking solvability of a given language $\Gamma$ is a decidable problem.
Indeed, condition (c) can be tested 
by solving
a linear program with $|\calO^{(4)}_D|=|D|^{|D|^4}$ variables and $O(poly(\size(\Gamma)))$ constraints, where we maximize the total weight of Siggers operations 
subject to linear constraints expressing that $\omega\in\mathbb R^{\calO^{(4)}_D}$ is a fractional polymorphism of $\Gamma$ of arity $4$.

\subsection{Basic LP relaxation}\label{sec:LP}

Symmetric operations are known to be closely related to LP-based algorithms for CSP-related problems.
One algorithm in particular has been known to solve many VCSPs to optimality. This algorithm is based
on the so-called {\em basic LP relaxation}, or BLP, defined as follows.

Let $\mathbb M_n=\{\mu\ge 0\:|\:\sum_{x\in D^n}\mu(x)=1\}$ be the set of probability distributions over
labelings in $D^n$.
We also denote $\Delta=\mathbb M_1$; thus, $\Delta$ is the standard ($|D|-1$)-dimensional simplex.
The corners of $\Delta$ can be identified with elements in $D$.
For a distribution $\mu\in\mathbb M_n$ and a variable $v\in\{1,\ldots,n\}$, let
 $\mu_{[v]}\in \Delta$ be the marginal probability of distribution $\mu$ for $v$:
\begin{equation*}
\mu_{[v]}(a)\ = \sum_{x\in D^n:x_v=a} \mu(x) \qquad \forall a \in D.
\end{equation*}
Given a VCSP instance $\calI$ in the form~\eqref{eq:VCSPinst}, we define the value $\BLP(\calI)$ as follows:
\begin{eqnarray}
&& \hspace{-85pt} \BLP(\calI)\ =\ \min_{\mu,\alpha}\ \sum_{t\in T}\sum_{x\in \dom f_t}\mu_t(x)f_t(x)  \label{eq:BLP} \\
 \mbox{s.t.~~} (\mu_t)_{[k]}&=&\alpha_{v(t,k)} \hspace{20pt} \forall t\in T,k\in\{1,\ldots,n_t\} \nonumber \\
\mu_t&\in&\mathbb M_{n_t}                      \hspace{29pt} \forall t\in T \nonumber \\
\mu_t(x)&=&0                                   \hspace{43pt} \forall t\in T,x\notin \dom f_t \nonumber \\
\alpha_v&\in&\Delta                            \hspace{40pt} \forall v\in V \nonumber
\end{eqnarray}
If there are no feasible solutions then $\BLP(\calI)=\infty$.
The objective function and all constraints in this system are linear, therefore this is a linear program. Its size is polynomial in $\size(\calI)$, so
$\BLP(\calI)$ can be found in time polynomial in $\size(\calI)$.

We say that BLP \emph{solves} $\calI$ if
$\BLP(\calI)=\min_{x\in D^n}f_\calI(x)$, and BLP
solves $\VCSP\Gamma$ if it solves all instances $\calI$ of $\VCSP\Gamma$.
The following results are known.
\begin{theorem}[\cite{kolmogorov15:power}]
\label{thm:BLP}
(a) BLP solves $\VCSP\Gamma$ if and only if $\Gamma$ admits a symmetric fractional polymorphism of every arity $m\ge 2$.
(b) If $\Gamma$ is finite-valued then BLP solves $\VCSP\Gamma$ if and only if $\Gamma$ 
admits a symmetric fractional polymorphism of arity $2$
(i.e.\ if it is solvable).
\end{theorem}

BLP relaxation also plays a key role for general-valued languages, as the following result shows.
Recall that $u_A$ for a subset  $A\subseteq D$ is the unary function $D\rightarrow \{0,\infty\}$
with $\dom u_A=A$. 
\begin{definition}
Consider instance $\calI$ with the set of variables $V$ and domain $D$. For node $v\in V$ denote
$D_v=\{a\in D\:|\:\exists x\in D^V\mbox{ s.t.\ }f_\calI(x)<\infty,x_v=a\}$.
We define $\Feasplus\calI$ and $\calI+\Feasplus\calI$ to be the instances
with variables $V$ and
 the following objective functions:
$$
f_{\Feasplus\calI}(x)=\sum_{v\in V} u_{D_v}(x_v)
\qquad\qquad
f_{\calI+\Feasplus\calI}(x)=f_\calI(x)+f_{\Feasplus\calI}(x)
$$
\end{definition}
It is easy to see that
$
f_{\calI}(x)=
f_{\calI+\Feasplus\calI}(x)
$ for any $x\in D^V$. However, the BLP relaxations of these two instances may differ.
\begin{theorem}[\cite{KKZ:SICOMP17}] \label{th:KKZ:SICOMP17}
If $\Gamma$ is solvable and $\calI$ is a $\Gamma$-instance then
 BLP solves $\calI+\Feasplus\calI$.
\end{theorem}

If $\Gamma$ is solvable and we know a core $B$ of $\Gamma$, then an optimal solution for every  $\Gamma$-instance can be found
by using the standard self-reducibility method, in which we repeatedly add unary terms of the form $u_a(x_v)$
to the instance for different $v\in V$ and $a\in B$ and check whether this changes the optimal value of the BLP
relaxation. A formal description of the method is given below.
(Notations  $\calI[B]$ and $\calI+u_a(x_v)$ should be self-explanatory;
in particular, the former is the instance obtained from $\calI$ by restricting
each term to domain $B$).

\begin{algorithm}[H]
\DontPrintSemicolon
compute value $LP^\ast\!=\!BLP(\calI\!+\!\Feasplus\calI)$, then update $\calI\leftarrow\calI[B]$ (or return $\varnothing$ if $LP^\ast=\infty$)\;
\For{each variable $v\in V$ in some order}
{
	\For{each label $a\in B$ in some order}
	{
		let $\calI'=\calI+u_a(x_v)$, and compute $LP'=BLP(\calI'+\Feas{\calI'})$ \;
		if $LP'\!=\!LP^\ast$  then update $\calI\!\leftarrow\!\calI'$ and go to line 2 (i.e.\ proceed with the next variable $v$)\!\!\!\!\!\! 
	}
	return {\tt FAIL}
}
return labeling $x^\ast\in B^V$ where $x^\ast_v$ equals the label $a$ for which term $u_a(x_v)$ has been added
\caption{{\tt LP-Probe}$(\calI,B)$. Input: instance $\calI$ with variables $V$ and domain $D$, set~$B\subseteq D$\:
Output: either a labeling $x^\ast\in\arg\min_{x\in D^V}f_\calI(x)$ with $x^\ast\in B^V$ or a flag in $\{\varnothing,{\tt FAIL}\}$\;
\label{alg:LPProbe}}
\end{algorithm}

\begin{lemma}\label{lemma:LPProbe}
(a) If {\tt LP-Probe}$(\calI,B)$ returns a labeling $x^\ast$ then 
$x^\ast\in\arg\min_{x\in D^V}f_\calI(x)$. \\
(b) If {\tt LP-Probe}$(\calI,B)$ returns $\varnothing$ then instance $\calI$ is infeasible. \\
(c) Suppose that $\calI$ is a $\Gamma$-instance where $\Gamma$ is solvable and $B\in\Bcore\Gamma$.
Then {\tt LP-Probe}$(\calI,B)\!\ne\!\mbox{\tt FAIL}$.
\end{lemma}
\begin{proof}
Part (a) holds by construction,
and part (b) can be derived from the following two facts
(which hold under the preconditions of part (b)):
\begin{itemize}
\item $\min_{x\in B^V}f_\calI(x)=\min_{x\in D^V}f_\calI(x)$ by Lemma~\ref{lemma:core-basics}(c). 
\item BLP solves all instances to which it is applied during the algorithm.
Indeed, by Lemma~\ref{lemma:solvability:equiv} the language $\Gamma'=\Gamma[B]\cup\{u_a\:|\:a\in B\}$
is solvable. The initial instance is a $\Gamma$-instance, and all instances in line 4 are $\Gamma'$-instances.
The claim now follows from Theorem~\ref{th:KKZ:SICOMP17}.
\end{itemize}
\end{proof}

\subsection{Meta-questions and uniform algorithms}
In the light of the previous discussion, it is natural to ask the following
questions about a given language $\Gamma$:
(i) Is $\Gamma$ solvable?
(ii) Is $\Gamma$ a core?
(iii) What is a core of $\Gamma$?
Such questions are usually called {\em meta-questions} or {\em meta-problems} in the literature.
For finite-valued languages their computational complexity has been studied in~\cite{tz16:jacm}.
\begin{theorem}[\cite{tz16:jacm}]\label{th:finite:NPhard}
Problems (i) and (ii) for $\{0,1\}$-valued languages are NP-complete and co-NP-complete, respectively.
\end{theorem}
\begin{theorem}[\cite{tz16:jacm}]\label{th:tz16:core-algorithm}
There is a polynomial-time algorithm that, given a core finite-valued language~$\Gamma$, either
finds a binary idempotent symmetric fractional polymorphism $\omega$ of $\Gamma$
with $|\supp(\omega)|=O(poly(\size(\Gamma)))$, or asserts that none exists.
\end{theorem}
For crisp languages the following hardness results are known.
\begin{theorem}[\cite{hell:92}]\label{th:crisp:NPhard-core}
Deciding whether a given crisp language $\Gamma$ with a single binary relation
is a core is a co-NP-complete problem. 
(Equivalently, testing whether a given directed graph
has a non-bijective homomorphism onto itself is an NP-complete
problem).
\end{theorem}
\begin{theorem}[\cite{ChenLarose:17}]\label{th:crisp:NPhard-Siggers}
Deciding whether a given crisp language $\Gamma$ with binary relations
is solvable is an NP-complete problem. 
\end{theorem}
It is still an open question whether an analogue of Theorem~\ref{th:tz16:core-algorithm}
holds for crisp languages, i.e.\ whether solvability of a given core crisp language $\Gamma$
can be tested in polynomial time. However, it is known~\cite{ChenLarose:17} that
the answer would be positive assuming the existence of a certain {\em uniform} polynomial-time algorithm for CSPs.
\begin{definition}\label{def:UniformAlg}
Let $\calF$ be a class of languages. A uniform polynomial-time algorithm
for $\calF$ is a polynomial-time algorithm that, for each input $(\Gamma, \calI)$ 
with $\Gamma\in\calF$ and $\calI\in\VCSP\Gamma$,
computes $\min_x f_\calI(x)$.
\end{definition}

\begin{theorem}[\cite{ChenLarose:17}]
Suppose that there exists a uniform polynomial-time algorithm for the class of solvable core crisp languages.
Then there exists a polynomial-time algorithm that decides whether a given core crisp language is
solvable (or equivalently admits a Siggers polymorphism).
\end{theorem}
Currently it is not known whether a uniform polynomial-time algorithm for core crisp languages exists.
(Algorithms in~\cite{Bulatov:FOCS17,Zhuk:FOCS17} assume that  needed polymorphisms of the language are part of the input;
furthermore, the worst-case bound on the runtime is exponential in $|D|$).

We remark that~\cite{ChenLarose:17} considered a wider range of meta-questions for crisp languages.
In particular, they studied the complexity of deciding whether
a given $\Gamma$ admits polymorphism $g\in O^{(m)}_D$ satisfying a given
{\em strong linear Maltsev condition} specified by a set of linear identities. Examples
of such identities are $g(x,\ldots,x)\approx x$ (meaning that $g$ is idempotent),
$g(x_1,x_2,\ldots,x_m)\approx g(x_2,\ldots,x_m,x_1)$ (meaning that $g$ is cyclic),
and $g(r,a,r,e)\approx  g(a,r,e,a)$ (meaning that $g$ is Siggers). 
We refer to~\cite{ChenLarose:17} for further details.


\section{Our results}\label{sec:results}
In this section the domain of language $\Gamma$ is always denoted as $D$, and its size as $d=|D|$.

Our algorithms will construct $\Gamma$-instances $\calI$ on $n=d^m$ variables (where $m\le 4$) with 
$\size(\calI)=O(poly(\size(\Gamma)))$ for some fixed polynomial.
We denote $\TFeas{n,\Gamma}$ to be the running time of a procedure that computes $\Feasplus\calI$
for such $\calI$'s.
Also, let $\TLPPROBE{n,\Gamma}$ be the combined running times of computing 
$\Feasplus\calJ$ for instances $\calJ$ during a call to {\tt LP-Probe}$(\calI,B)$
for such $\calI$ and some subset $B$.
Note, if $\Gamma$ is finite-valued then computing $\Feasplus\calI$ is a trivial problem,
so $\TFeas{n,\Gamma}$ and $\TLPPROBE{n,\Gamma}$ would be polynomial in $n+\size(\Gamma)$.

\paragraph{Conditional cores} First, we consider the problem of computing a core $B\in\Bcore\Gamma$ of a given
language $\Gamma$. A naive solution is to solve a linear program with $|\calO^{(1)}|=d^{d}$ variables.
We will present an alternative technique that runs more efficiently (in the case of finite-valued languages)
but is allowed to output an incorrect result if $\Gamma$ is not solvable.
It will be convenient to introduce the following terminology:
language $\Gamma$ is a {\em conditional core} if either $\Gamma$ is a core or $\Gamma$ is not solvable.
Similarly, set $B$ is a {\em conditional core of $\Gamma$} if either $B\in\Bcore\Gamma$ or $\Gamma$ is not solvable.
Note, $B=\varnothing$ is a conditional core of $\Gamma$ if and only if $\Gamma$ is not solvable.

To compute a conditional core of $\Gamma$, we will use the following approach.
Consider a pair $(\Gamma,\sigma)$ where  $\sigma$ is a string of size $O(poly(\Gamma))$ that specifies set $\calB_\sigma$ of candidate cores of $\Gamma$.
Formally,  $\calB_\sigma=\{B_1,\ldots,B_N\}$ where $\varnothing\ne B_i\subseteq D$ for each $i\in[N]$.
We assume
that elements of $\calB_\sigma$ can be efficiently enumerated,
i.e.\ there exists
a polynomial-time procedure for computing $B_1$ from $\sigma$ and $B_{i+1}$ from $(\sigma,B_i)$.
If $\calB$ is a set of subsets $B\subseteq D$, we will denote
\begin{eqnarray*}
\calO[\calB]&=&\{g\in\calO^{(1)}\:|\:g(D)= B\mbox{ for some }B\in\calB\} \\
\calOplus[\calB]&=&\{g\in\calO^{(1)}\:|\:g(D)\subseteq B\mbox{ for some }B\in\calB\} 
\end{eqnarray*}
\begin{theorem}\label{th:find-core:general}
There exists an algorithm 
that for a given  input $(\Gamma,\sigma)$ does one of the following:
\begin{itemize}
\item[(a)] Produces a fractional polymorphism $\omega\in\fPol[1]\Gamma$ with
$\supp(\omega)\subseteq\calOplus[\calB_\sigma]$ and 
$|\supp(\omega)|\le 1+\sum_{f\in\Gamma}|\dom f|$.
\item[(b)] Asserts that there exists no vector $\omega\in\fPol[1]\Gamma$ with $\supp(\omega)\subseteq\calOplus[\calB_\sigma]$.
\item[(c)] Asserts that one of the following holds: (i) $\Gamma$ is not solvable; (ii) $\calB_\sigma\cap\Bcore\Gamma=\varnothing$.
\end{itemize} 
It runs in $(|\calB_\sigma|\!+\!O(poly(\size(\Gamma))))\!\cdot\! (\TLPPROBE{d,\Gamma}\!+\!O(poly(\size(\Gamma))))$
time and uses $O(poly(\size(\Gamma)))$ space.\!\!\!\!\!
\end{theorem}
The algorithm in the theorem above is based on the ellipsoid method~\cite{GLS88:book},
which tests feasibility of a polytope using a polynomial number of calls to the separation oracle.
In our case this oracle is implemented via one or more calls to {\tt LP-Probe}$(\calI,B)$
for appropriate $\calI$ and $B$.

One possibility would be to use Theorem~\ref{th:find-core:general} with the set $\mbox{$\calB_\sigma=\{B\subseteq D\:|\:B\ne\varnothing,B\ne D\}$}$.
If the algorithm returns result (a) then we can take operation $g\in\supp(\omega)$ and call the algorithm recursively for the language $\Gamma[g(D)]$
on a smaller domain. If we get result (b) or (c) then one can show that $\Gamma$ is a conditional core, so we can stop.
For finite-valued languages this approach would run in $O(2^{d}\cdot poly(\size(\Gamma)))$ time.
We will pursue an alternative approach with an improved complexity
 $O(\sqrt[3]3^{\,d}\cdot poly(\size(\Gamma)))$.

This approach will use partitions $\Pi=\{D_1,\ldots,D_k\}$ of domain~$D$. For such $\Pi$
we denote 
\begin{eqnarray*}
\calO_\Pi&=&\{g\in\calO^{(1)}_D\::\:g(a)=g(b)~~~~~~\,\quad\forall a,b\in A\in\Pi\} \\
\BcorePi&=&\{B\subseteq D\::\: |B\cap A|=1~~~~~~\hspace{6.5pt}\quad\forall A\in\Pi\} 
\end{eqnarray*}
We say that $\Pi$ is a {\em partition of $\Gamma$} if the set $\calO_\Pi\cap\fPolplus\Gamma$ is non-empty.
In particular, the partition  $\Pi=\{\{a\}\:|\:a\in D\}$
of $D$ into singletons is a partition of $\Gamma$, since $\fPolplus\Gamma$
contains the identity mapping  $D\rightarrow D$.
We say that $\Pi$ is a {\em maximal partition of $\Gamma$} if $\Pi$ is a partition of $\Gamma$ and no coarser partition $\Pi'\succ\Pi$ 
(i.e.\  $\Pi'$ with $\calO_{\Pi'}\subset \calO_\Pi$) has this property. Clearly, for any $\Gamma$ there exists
at least one $\Pi$ which is a maximal partition of $\Gamma$. 
By analogy with cores, we say that $\Pi$ is a {\em conditional (maximal) partition of $\Gamma$}
if either $\Pi$ is a (maximal) partition of $\Gamma$ or $\Gamma$ is not solvable.

In the results below 
$\Pi$ is always assumed to be a partition of~$D$. 
\begin{lemma}\label{lemma:Pi}
(a) If $\Pi$ is a maximal partition of  $\Gamma$ then $\Bcore\Gamma\subseteq\BcorePi$ and  $\Ocore\Gamma=\OcorePi\cap\fPolplus\Gamma$. \\
(b) If $\Pi$ is a partition of $D$ then $|\BcorePi|\le \sqrt[3]3^{\,d}$.
\end{lemma}
\begin{theorem}\label{th:Pi-testing}
There exists an algorithm with runtime $\TFeas{d,\Gamma}+\TFeas{|\Pi|,\Gamma}+O(poly(\size(\Gamma)))$
that for a given input $(\Gamma,\Pi)$  does one of the following:
\begin{itemize}
\item[(a)] Asserts that $\Pi$ is a conditional partition of $\Gamma$.
\item[(b)] Asserts that $\Pi$ is not a partition of $\Gamma$.
\end{itemize} 
\end{theorem}
As before, the algorithm in Theorem~\ref{th:Pi-testing} is based on the ellipsoid method.
However, now we cannot use procedure {\tt LP-Probe}$(\calI,B)$
to implement the separation oracle, since a candidate core $B$ is not available.
Instead, we solve the BLP relaxation of instance $\calI$ and derive a separating hyperplane
from a (fractional) optimal solution of the relaxation.

\begin{corollary}\label{cor:gstar}
{\em (1)} A conditional maximal partition $\Pi$ of $\Gamma$ 
can be computed in $O(d^2)\cdot \TFeas{d,\Gamma}+O(poly(\size(\Gamma))$ time.
{\em (2)} Once such $\Pi$ is found,
a conditional core $B$ of $\Gamma$ 
can be computed using
$(|\BcorePi|+O(poly(\size(\Gamma)))\cdot (\TLPPROBE{d,\Gamma}+O(poly(\size(\Gamma))))$
time and $O(poly(\size(\Gamma)))$ space.
If $B\ne\varnothing$ then the algorithm also produces
a fractional polymorphism $\omega\in\fPol\Gamma$ such
that $\supp(\omega)\subseteq \OcorePi$, $|\supp(\omega)|\le 1+\sum_{f\in\Gamma}|\dom f|$ and
$\supp(\omega)$ contains an operation $g$ with $g(D)=B$.
\end{corollary}
In part (1) we use a greedy search that starts with $\Pi=\{\{a\}\:|\:a\in D\}$ and then
repeatedly calls the algorithm in Theorem~\ref{th:Pi-testing} for coarser partitions $\Pi$.
In part (2) we call the algorithm from Theorem~\ref{th:find-core:general} with $\sigma=\Pi$ and $\calB_\sigma=\BcorePi$.
For further details we refer to Appendix~\ref{sec:core:gstar:proof}.

\paragraph{Testing solvability of a conditional core} Once we have found a conditional core $B$ of $\Gamma$,
we need to test whether  language $\Gamma[B]$ is solvable. This problem is known
to be solvable in polynomial-time for finite-valued languages~\cite{tz16:jacm},
see Theorem~\ref{th:tz16:core-algorithm}. Their result can be extended as follows.

\begin{theorem}\label{th:24}
There exists an algorithm that for a given language $\Gamma$ does one of the following:
\begin{itemize}
\item[(a)] Produces an idempotent fractional polymorphism $\omega\in\fPol\Gamma$ certifying solvability of $\Gamma$: \\
$\bullet$ $\omega$ has arity $m=2$ and is symmetric, if $\Gamma$ is finite-valued; \\ 
$\bullet$ $\omega$ has arity $m=4$ and contains a Siggers operation in the support, if $\Gamma$ is not finite-valued.
Furthermore, in each case vector $\omega$ satisfies $|\supp(\omega)|\le 1+\sum_{f\in\Gamma}\binom{|\dom f|}{m}$.
\item[(b)] Asserts that one of the following holds: (i) $\Gamma$ is not solvable; (ii) $\Gamma$ is not a core.
\end{itemize} 
Its runtime is $O(poly(\size(\Gamma)))$  if $\Gamma$ is finite-valued, and 
$O(\TLPPROBE{d^4,\Gamma}\cdot poly(\size(\Gamma)))$ otherwise.
\end{theorem}

Combining procedures in Corollary~\ref{cor:gstar} 
and the algorithm in Theorem~\ref{th:24}
yields our main algorithmic result.
\begin{corollary}
Solvability of a given finite-valued language $\Gamma$ can be tested
in $O(\sqrt[3]{3}^{\,d}\cdot poly(\size(\Gamma)))$ time.
If the answer is positive, the algorithm also returns a fractional polymorphism $\omega_1\in\fPol[1]\Gamma$ with $\supp(\omega_1)\subseteq\Ocore\Gamma$
and a symmetric idempotent fractional polymorphism $\omega_2\in\fPol[2]{\Gamma[B]}$ where $B\!=\!g(D)$ for some $g\!\in\!\supp(\omega_1)$;
furthermore, $|\supp(\omega_m)|\!\le\! 1\!+\!\sum_{f\in\Gamma}\binom{|\dom f|}{m}$ for $m\!\in\!\{1,2\}$.\!\!\!\!\!\!
\end{corollary}

\paragraph{Hardness results} Let us fix a constant $L\in\{1,\infty\}$.
As Theorems~\ref{th:finite:NPhard},~\ref{th:crisp:NPhard-core} and~\ref{th:crisp:NPhard-Siggers} state,
testing whether $\Gamma$ is (i) solvable and (ii) is a core are both NP-hard problems for  $\{0,L\}$-valued languages.
We now present additional hardness results under the Exponential Time Hypothesis (ETH) and the Strong Exponential Time Hypothesis (SETH)
(see Conjectures~\ref{conj:ETH} and~\ref{conj:SETH}).
Note that for Theorem~\ref{th:ETH-hardness} we simply reuse the reductions from~\cite{ChenLarose:17}.
We say that a family of languages $\calF$ is  {\em $k$-bounded}
if each $\Gamma\in\calF$ satisfies $\size(\Gamma)=O(poly(d))$ for some fixed polynomial,
and ${\tt arity}(f)\le k$ for all $f\in\Gamma$.

\begin{theorem}\label{th:ETH-hardness}
Suppose that ETH holds. Then there exists a 2-bounded family $\calF$ of $\{0,L\}$-valued languages
such that the following problems cannot be solved in $O(2^{o(d)})$ time: \\
(a) Deciding whether language  $\Gamma\in\calF$ is solvable. \\
(b) Deciding whether language $\Gamma\in\calF$ is a core.  
\end{theorem}

\begin{theorem}\label{th:SETH-hardness}
Suppose that SETH holds. Then for any  $\delta<1$ 
there exists an $O(1)$-bounded family $\calF$ of $\{0,L\}$-valued languages
such that the following problems cannot be solved in $O(\sqrt[3]3^{\,\delta d})$ time: \\
(a) Deciding whether language  $\Gamma\in\calF$ is solvable (assuming the existence of a uniform polynomial-time
algorithm for core crisp languages, in the case when $L=\infty$). \\
(b) Deciding whether language  $\Gamma\in\calF$ satisfies $\coresize(\Gamma)\le d/3$.  
\end{theorem}




\section{Proofs}\label{sec:proofs}
\subsection{Ellipsoid method}
Using the ellipsoid method,
 Gr\"otschel, Lov\'asz and Schrijver~\cite{GLS88:book} established a polynomial-time equivalence between
linear optimization and separation problems in polytopes. We will need
one implication of this result, namely that efficient separation implies
efficient feasibility testing. A formal statement is given below.

Consider a family of instances where an instance $\Lambda$ can be described by a string of length $\size(\Lambda)$
over a fixed alphabet.
Suppose that for each $\Lambda$ we have an integer $n$ and a finite set $\calG$, where each element $g\in\calG$
corresponds to a hyperplane $c_g\in\mathbb Q^{n+1}$.
This hyperplane encodes linear inequality $\langle c_g,[y\; 1]\rangle \ge 0$ on vector $y\in\mathbb R^n$.
Let us denote $\HP{g}=\{y\in\mathbb R^n\:|\: \langle c_g,[y\; 1]\rangle \ge 0\}$ and 
$\HP{\calG'}=\bigcap_{g\in\calG'}\HP{g}$ for a subset $\calG'\subseteq\calG$.

We make the following assumptions:
(i) each $g\in\calG$ can be described by a string of size $O(poly(\size(\Lambda))$;
(ii)  vector $c_g$ can be computed from $\Lambda$ and $g$ in polynomial time
(implying that $\size(c_g)=O(poly(\size(\Lambda)))$, where the size of a vector in $\mathbb Q^{n+1}$ is the sum of sizes of its $n+1$ components);
(iii) set  $\calG$ can be constructed algorithmically from input  $\Lambda$.
Note that quantities $n$, $\calG$, $\{(c_g,\HP{g})\:|\:g\in\calG\}$ all depend on~$\Lambda$;
for brevity this dependence is not reflected in the notation.
\begin{theorem}[{\cite[Lemma 6.5.15]{GLS88:book}}]
Consider the following problems:
\em
\begin{itemize}
\item  {[{\tt Separation}]}  Given instance $\Lambda$ and vector $y\in \mathbb Q^n$, either decide that $y\in \HP{\calG}$, or find
a separating hyperplane $c\in\mathbb Q^{n+1}$ with $\size(c)=O(poly(\size(\Lambda)+\size(y)))$ satisfying  $\langle c,[y\; 1]\rangle < 0$ and $\langle c,[z\; 1]\rangle \ge 0$ for all $z\in\HP{\calG}$. 
\item   {[{\tt Feasibility}]}  Given instance $\Lambda$, decide whether $\HP{\calG}=\varnothing$. 
\end{itemize}
\em
There exists an algorithm for solving {\em[{\tt Feasibility}]} that makes a polynomial number
of calls to the oracle for {\em[{\tt Separation}]} plus a polynomial number of other operations.
\label{th:ellipsoid}
\end{theorem}
Note that a (possibly inefficient) algorithm for solving [{\tt Separation}] always exists:
if $y\notin\HP\calG$ then one possibility is to find an element $g\in\calG$ with $y\notin\HP{g}$
and return hyperplane $c_g$. (Its size is polynomial in $\size(\Lambda)$ by assumption).

For some parts of the proof we will also need the following variation.
\begin{theorem}
Consider the following problems:
\em
\begin{itemize}
\item  {[{\tt Separation+}]}  Given instance $\Lambda$ and vector $y\in \mathbb Q^n$, either decide that $y\in \HP{\calG}$, or find
an element $g\in\calG$ with $y\notin\HP{g}$ (i.e.\ an element $g$ with  $\langle c_g,[y\; 1]\rangle < 0$).
\item   {[{\tt Feasibility+}]}  Given instance $\Lambda$, decide whether $\HP{\calG}=\varnothing$. If $\HP{\calG}=\varnothing$, find a subset $\calG'\subseteq\calG$
such that $|\calG'|=O(poly(\size(\Lambda)))$ and $\HP{\calG'}=\varnothing$.
\end{itemize}
\em
There exists an algorithm for solving {\em[{\tt Feasibility+}]} that makes a polynomial number
of calls to the oracle for {\em[{\tt Separation+}]} plus a polynomial number of other operations.
\label{th:ellipsoid:plus}
\end{theorem}
This result can be deduced from Theorem~\ref{th:ellipsoid}:
the desired subset $\calG'$ can be taken as the set of all elements in $\calG$ returned by the oracle during the algorithm.


\subsection{Farkas lemma for fractional polymorphisms}
Let us fix integer $m\ge 1$ and sets $\Omega^-,\Omega^+$ with $\Omega^-\subseteq\Omega^+\subseteq \Pol\Gamma\cap\calO^{(m)}$.
These choices will be specified later (they will depend on the specific theorem that we will be proving).
Let $\Gamma^+$ be the set of tuples $(f,\bx)$ such that $f:D^n\rightarrow \Qc$ is an $n$-ary function in $\Gamma$ and
$\bx\in[\dom f]^m$. Note, $\bx$ can be viewed as a matrix of size $m \times n$:
\begin{equation*}
\bx=\left(\begin{tabular}{ccc}
$x^1_1$ & $\ldots$ & $x^1_n$ \\
$\ldots$ & $\ldots$ & $\ldots$ \\
$x^m_1$ & $\ldots$ & $x^m_n$
\end{tabular}
\right)
\end{equation*}
For such $\bx$ we will write $x^i=(x^i_1,\ldots,x^i_n)\in D^n$ and $x_j=(x^1_j,\ldots,x^m_j)\in D^m$.
For an operation $g\in\calO^{(m)}$ we denote
$g(\bx)=(g(x_1),\ldots,g(x_n))\in D^n$,
and for a cost function $f:D^n\rightarrow \Qc$ we denote
$f^m(\bx)=\frac{1}{m}(f(x^1)+\ldots+f(x^m))$.

Next, we define various hyperplanes in $\mathbb Q^{\Gamma^+}\cup\mathbb Q$ as follows:
\begin{itemize}
\item For $g\in\Omega^+$ let $c_g$ be the hyperplane
corresponding to the inequality
\begin{eqnarray}
\sum_{(f,\bxS)\in\Gamma^+}[f(g(\bx))-f^{m}(\bx)]\cdot y(f,\bx) & \ge & [g\in\Omega^-] \label{eq:g:inequality}
\end{eqnarray}
where we used the Iverson bracket notation: $[\phi]=1$ if $\phi$ is true, and $[\phi]=0$ otherwise.
\item For $(f,\bx)\in\Gamma^+$ let $c_{(f,\bxS)}$ be the hyperplane
corresponding to the inequality $y(f,\bx)\ge 0$.
\item Introduce a special element $\perp$, and let $c_\perp=(0,\ldots,0,1)$ be the hyperplane
corresponding to the (unsatisfiable) inequality $0 \ge 1$.
\end{itemize}

For a subset $\Omega\subseteq\Omega^+$ it will be convenient to denote $Y[\Omega]=\HP{\Omega\cup\Gamma^+}$.
In other words, $Y[\Omega]$ is the set of vectors $y\in\mathbb R^{\Gamma^+}_{\ge 0}$ satisfying
\begin{eqnarray}
\sum_{(f,\bxS)\in\Gamma^+}[f(g(\bx))-f^{m}(\bx)]\cdot y(f,\bx) & \ge & [g\in\Omega^-] \qquad\qquad\forall g\in\Omega \label{eq:OmegaSystem:a}
\end{eqnarray}

\begin{lemma}\label{lemma:Farkas}
Suppose that $\Omega\subseteq\Omega^+$. Then $Y[\Omega]=\varnothing$ if and only if
$\Gamma$ admits an $m$-ary fractional polymorphism $\omega$ 
such that $\supp(\omega)\subseteq\Omega$ and $\supp(\omega)\cap\Omega^-\ne\varnothing$.

If $Y[\Omega]=\varnothing$ then it is possible to compute such $\omega$ in $O(poly(\size(\Gamma)+|\Omega|))$ time (given~$\Gamma$ and~$\Omega$)
so that it additionally satisfies $|\supp(\omega)|\le 1+|\Gamma^+|=1+\sum_{f\in\Gamma}\binom{|\dom f|}{m}$.
\end{lemma}
\begin{proof}
Introducing slack variables $\{y(g)\:|\:g\in \Omega\}$, we have $Y[\Omega]=\varnothing$ if and only if the following system
does not have a solution~$\mbox{$y\in\mathbb R^{\Omega\cup\Gamma^+}_{\ge 0}$}$:
\begin{eqnarray}
y(g) - \sum_{(f,\bxS)\in\Gamma^+}[f(g(\bx))-f^m(\bx)]\cdot y(f,\bx) & = & -[g\in\Omega^-] \qquad\qquad\forall g\in\Omega 
\end{eqnarray}
By Farkas Lemma this happens if and only if the following system has a solution $\omega\in\mathbb R^{\Omega}$:
\begin{subequations}\label{eq:AGALKSFJAKSG}
\begin{eqnarray}
-\sum_{g\in\Omega}[f(g(\bx))-f^m(\bx)]\cdot \omega(g) & \ge & 0 \qquad\qquad\forall (f,\bx)\in\Gamma^+ \\
\omega(g)                                                     & \ge & 0 \qquad\qquad\forall g\in\Omega \\
-\sum_{g\in\Omega\cap\Omega^-}\omega(g)                                     & < & 0 
\end{eqnarray}
\end{subequations}
If a solution exists, then it can be chosen to satisfy $\sum_{g\in\Omega}\omega(g)=1$
(since multiplying feasible solutions to~\eqref{eq:AGALKSFJAKSG} by a positive constant gives feasible solutions). We obtain that $Y[\Omega]=\varnothing$
if and only if there exists vector $\omega\in\mathbb R^\Omega$ satisfying
\begin{subequations}\label{eq:Farkas:proof}
\begin{eqnarray}
\sum_{g\in\Omega}\omega(g) f(g(\bx))  & \le & f^m(\bx) \qquad\qquad\forall (f,\bx)\in\Gamma^+ \label{eq:Farkas:proof:a} \\
\omega(g)                             & \ge & 0 \hspace{25.5pt} \qquad\qquad\forall g\in\Omega \label{eq:Farkas:proof:b} \\
\sum_{g\in\Omega}\omega(g)   & = & 1 \label{eq:Farkas:proof:c} \\
\sum_{g\in\Omega\cap\Omega^-}\omega(g)                                     & > & 0 
\end{eqnarray}
\end{subequations}
This establishes the first claim of Lemma~\ref{lemma:Farkas}. Now suppose that $Y[\Omega]=\varnothing$,
so that system~\eqref{eq:Farkas:proof} has a solution.
Inequalities~\eqref{eq:Farkas:proof:a}-\eqref{eq:Farkas:proof:c} can be written in a matrix form as $A\omega\le b$
with $A\in\mathbb R^{m\times n}$, $b\in\mathbb R^m$ where   $n=|\Omega|$ is the number of variables
and $m=|\Gamma^+|+|\Omega|+2$ is the number of constraints. (The equality~\eqref{eq:Farkas:proof:c} is represented as two inequalities)
Let $P$ be the polytope $\{\omega\in\mathbb R^n\:|\:A\omega\le b\}$,
and $\omega$ be a vertex of $P$ that maximizes 
 $\sum_{g\in\Omega\cap\Omega^-}\omega(g)$. (By standard results from linear programming, such $\omega$ can be 
computed in polynomial time). It now suffices to show that $|\supp(\omega)|\le |\Gamma^+|+1$.

We have $rank(A)=n$, since $-A$ contains the identity submatrix of size $n\times n$.
Since $\omega$ is a vertex of $P$, $A$ has a non-singular submatrix $A'\in\mathbb R^{n\times n}$ such that $A'\omega=b'$,
where $b'$ is the corresponding subvector of $b$. Matrix $A'$ has at most $|\Gamma^+|+1$ rows corresponding
to constraints~\eqref{eq:Farkas:proof:a} and~\eqref{eq:Farkas:proof:b} (note that the two rows of $A$ corresponding to constraint~\eqref{eq:Farkas:proof:b}
are linearly dependent and thus cannot be both present in $A'$). Thus, $A'$ has at least $n-|\Gamma^+|-1$ rows
corresponding to constraints~\eqref{eq:Farkas:proof:c}. These constraints are tight for $\omega$,
and so at least $n-|\Gamma^+|-1$ components of $\omega$ are zeros. Thus, $\omega$ has at
most $n - (n-|\Gamma^+|-1)=|\Gamma^+|+1$ non-zero components.
\end{proof}

\subsection{Proof of Theorem~\ref{th:find-core:general} (enumeration of candidate cores)}\label{sec:th:find-core:general:proof}
The input to the desired algorithm is a pair $\Lambda=(\Gamma,\sigma)$.
Let us define $m=1$ and 
$
\Omega^-=\Omega^+=\calOplus[\calB_\sigma]\cap\Pol\Gamma
$.
Note that condition ``$\omega\in\fPol\Gamma$ and $\supp(\omega)\subseteq\calOplus[\calB_\sigma]$'' used in Theorem~\ref{th:find-core:general}(a,b)
is equivalent to the condition ``$\omega\in\fPol\Gamma$ and $\supp(\omega)\subseteq\Omega^+$'', since $\fPolplus\Gamma\subseteq\Pol\Gamma$.
By Lemma~\ref{lemma:Farkas}, vector $\omega$ satisfying these conditions exists if and only if $Y[\Omega^+]=\varnothing$.

\begin{lemma}\label{lemma:find-core:general:separation}
There exists an algorithm with complexity $\TLPPROBE{|D|,\Gamma}+O(poly(\size(\Gamma)+\size(y)))$ that given a tuple $(\Gamma,B)$ with $B\subseteq D$ 
and a vector $y\in\mathbb Q^{\Gamma^+}$ does one of the following:
\begin{itemize}
\item[(a)] Produces element $g\in\Omega^+\cup\Gamma^+$ such that $\langle c_g,[y\; 1]\rangle<0$.
\item[(b)] Asserts that one of the following holds: (i) $\Gamma$ is not solvable; (ii) $B\notin\Bcore\Gamma$. 
\end{itemize}
\end{lemma}
Before proving this lemma, let us describe how it implies Theorem~\ref{th:find-core:general}.
We use the ellipsoid method where for the input $\Lambda=(\Gamma,\sigma)$ 
we define $\calG=\Omega^+\cup\Gamma^+$. 
Theorem~\ref{th:ellipsoid:plus} gives a class of algorithms for solving [{\tt Feasibility+}]
(that depend on the implementation of separation oracles).
If $\HP\calG=Y[\Omega^+]=\varnothing$ then we have a subset $\calG'\subseteq \calG$ of polynomial size with $\HP{\calG'}=0$.
Let $\Omega'=\calG'\cap\Omega^+\subseteq\Omega^+$, then $Y[\Omega']=\varnothing$.
Using Lemma~\ref{lemma:Farkas}, we construct a unary fractional polymorphism $\omega$ of $\Gamma$
with $\supp(\omega)\subseteq\Omega'\subseteq\Omega^+$ and $|\supp(\omega)|\le 1+\sum_{f\in\Gamma}|\dom f|$,
and return it as the result.
If $\HP\calG=Y[\Omega^+]\ne\varnothing$ then we can output result (b) in Theorem~\ref{th:find-core:general}
(by the observation in the beginning of this section).
We obtained a correct but possibly inefficient algorithm for solving the problem in Theorem~\ref{th:find-core:general}.
Next, we will modify this algorithm so that it remains correct and has the desired complexity.

We will  maintain a subset $\calB=\{B_1,\ldots,B_k\}\subset \calB_\sigma$ with $0\le k<N$ (initialized with $B=\varnothing$) that satisfies the following invariant:
\begin{itemize}
\item[($\star$)] If $\Gamma$ is solvable then $\calB\cap\Bcore\Gamma=\varnothing$.
\end{itemize}

When the algorithm calls the oracle for [{\tt Separation+}] with vector $y\in\mathbb Q^{\Gamma^+}$,
we do the following:
\begin{enumerate}
\item Pick $B=B_{k+1}\in\calB_\sigma-\calB$ and call the algorithm from Lemma~\ref{lemma:find-core:general:separation} with the input $(\Gamma,B)$.
\item If it returns $g\in\Omega^+\cup\Gamma^+$ with $\langle c_g,[y\; 1]\rangle<0$ then return $g$ as the output of the oracle.
\item If it asserts that $\Gamma$ is not solvable or $B\notin\Bcore\Gamma$ 
then  add $B$ to $\calB$; clearly, this preserves invariant ($\star$). If $\calB\ne\calB_\sigma$ then go to step 1.
If we get $\calB=\calB_\sigma$  then from ($\star$) we conclude that either $\Gamma$ is not solvable or $\calB_\sigma\cap\Bcore\Gamma=\varnothing$.
Thus, we can output result (c) in Theorem~\ref{th:find-core:general} and terminate.
\end{enumerate}

We claim that this algorithm has the complexity stated in Theorem~\ref{th:find-core:general}. Indeed, each call to the oracle for [{\tt Separation+}]
can involve several passes of steps 1-3. Each pass has complexity $\TLPPROBE{|D|,\Gamma}+O(poly(\size(\Gamma)))$
(note that $\size(y)$ would be polynomial in $\size(\Gamma)$).
Call the last pass {\em successful},
and the other passes {\em unsuccessful}.
The number of successful passes is polynomial in $\size(\Gamma)$ (since the number of oracle calls is polynomial),
and the number of unsuccessful passes is at most~$|\calB_\sigma|$. This establishes the claim about the complexity.

\paragraph{Proof of Lemma~\ref{lemma:find-core:general:separation}}
We can assume that the input vector $y\in\mathbb Q^{\Gamma^+}$ is non-negative
(otherwise we can easily find element $(f,\bx)\in\Gamma^+$ with $\langle c_{(f,\bxS)},[y \;1]\rangle<0$). 
Let $\calI(y)$ be the instance with variables $D$ and the following objective function:
\begin{eqnarray*}
f_{\calI(y)}(g)&=&\sum_{(f,\bxS)\in\Gamma^+} y(f,\bx)f(g_{x_1},\ldots,g_{x_n}) \qquad\quad\hspace{11pt}\forall g:D\rightarrow D 
\end{eqnarray*}
Note, we wrote $g_a$ instead of $g(a)$ to emphasize that $g$ is now treated as a labeling of instance~$\calI(y)$,
though mathematically it is the same object as an operation in $\calO^{(1)}$.
(This is the convention that we use for instances, see eq.~\eqref{eq:VCSPinst}).

Also note that we  are now using rational weights in the definition of an instance,
where we adopt the following convention: if the weight $y(f,\bx)$ is zero
then expression $y(f,\bx)f(\ldots)$ means $\dom f(\ldots)$.
To be consistent with the original definition, one could multiply everything by a constant to make all weights integers,
and then treat these integers as the number of occurrences of the corresponding terms. However, this would make the notation cumbersome,
so we avoid this. 

The following equation can be easily verified:
\begin{eqnarray}
f_{\calI(y)}(g)&=&
\begin{cases}
  \mathlarger\sum\limits_{(f,\bxS)\in\Gamma^+} y(f,\bx)f(g(\bx))\;<\;\infty \qquad & \mbox{if } g\in\calO^{(1)}\cap \Pol\Gamma  \\
  \;\;\;\;\infty & \mbox{if } g\in\calO^{(1)} - \Pol\Gamma 
\end{cases}
\end{eqnarray}
In particular, we have $f_{\calI(y)}({\mathds 1})<\infty$ where ${\mathds 1}\in\Pol\Gamma$ is the identity mapping $D\rightarrow D$. 

Let us run procedure {\tt LP-Probe}$(\calI(y),B)$.
It cannot return $\varnothing$ since $\calI(y)$ has at least one feasible solution.
If it returns {\tt FAIL} then by Lemma~\ref{lemma:LPProbe} either $\Gamma$ is not solvable or $B\notin\Bcore\Gamma$.
Thus, we can output result (b) in Lemma~\ref{lemma:find-core:general:separation}.
Now suppose that {\tt LP-Probe}$(\calI(y),B)$ returns labeling $g:D\rightarrow B$.
We will show next that $g\in\Omega^+$ and $\langle c_g,[y\;1]\rangle<0$, and thus element $g$
can be returned as the output of the algorithm in Lemma~\ref{lemma:find-core:general:separation}.


We have $f_{\calI(y)}(g)\le f_{\calI(y)}({\mathds 1})<\infty$, since $g$ is a minimizer of $f_{\calI(y)}$.
Thus, $g\in\Pol\Gamma$. By construction, $g(D)\subseteq B\in\calB_\sigma$, and so $g\in \calOplus[\calB_\sigma]\cap\Pol\Gamma=\Omega^+$. We can now prove the claim:
$$
\langle c_g,[y\; 1]\rangle 
\;\;=\;\; f_{\calI(y)}(g) - f_{\calI(y)}({\mathds 1}) - 1 
\;\;\le\;\; -1 
\;\;<\;\; 0
$$

\subsection{Proof of Theorem~\ref{th:Pi-testing} (testing partition $\Pi$ of $D$)}\label{sec:Pi-testing:proof}
The general structure of the proof will be same as in the previous section, but some details will differ.
In particular, we will not be able to apply procedure {\tt LP-Probe}$(\calI,B)$ since we do not know a core $B$,
and will not be able to obtain an integer solution of the BLP relaxation. Instead, we will derive
a separating hyperplane from an optimal (fractional) solution of the BLP relaxation.

The input to the desired algorithm is a pair $(\Gamma,\Pi)$ where $\Pi$ is a partition of $D=D_\Gamma$.
We set $m=1$, $\Omega^-=\calO_\Pi\cap\Pol\Gamma$ and $\Omega^+=\calO^{(1)}\cap\Pol\Gamma$.
Using Lemma~\ref{lemma:Farkas} and the fact that $\fPolplus\Gamma\subseteq\Pol\Gamma$, we conclude that  $\Pi$ is a partition of $\Gamma$ if and only if $Y[\Omega^+]= \varnothing$. 

For an element $a\in D$ let $[a]$ be the unique set $A\in\Pi$ containing $a$.
Let $\calJ_\circ$ and $\calJ'_\circ$ be the instances with variables $D$ and $\Pi$, respectively, that have the following objective functions:
\begin{eqnarray*}
f_{\calJ_\circ}(g)&=&\sum_{(f,\bxS)\in\Gamma^+}\dom f(g_{x_1},\ldots,g_{x_n})\hspace{9pt}\qquad\quad\forall g:D\rightarrow D \\
f_{\calJ'_\circ}(g)&=&\sum_{(f,\bxS)\in\Gamma^+}\dom f(g_{[x_1]},\ldots,g_{[x_n]})\qquad\quad\forall g:\Pi\rightarrow D 
\end{eqnarray*}
We denote $\calJ= \Feasplus{\calJ_\circ}$ and $\calJ'=\Feasplus{\calJ'_\circ}$.
\begin{lemma}\label{lemma:Ph-testing:separation}
There exists a polynomial-time algorithm that given a tuple $(\Gamma,\Pi,\calJ,\calJ')$ 
and a vector $y\in\mathbb Q^{\Gamma^+}$ does one of the following:
\begin{itemize}
\item[(a)] Asserts that $y\in Y[\Omega^+]$.
\item[(b)]  Produces hyperplane $c\in\mathbb Q^{\Gamma^+}\times\mathbb Q$ such that $\langle c,[y\; 1]\rangle<0$
and one of the following holds: \\
(i) $\langle c,[z\; 1]\rangle\ge 0$ for all $z\in Y[\Omega^+]$; (ii) $\Gamma$ is not solvable.
\end{itemize}
\end{lemma}
Before proving this lemma, let us describe how it implies Theorem~\ref{th:Pi-testing}.
We use the ellipsoid method with the input $\Lambda=(\Gamma,\Pi)$ 
where the set~$\calG$ is defined as follows: if $\Gamma$ is solvable then $\calG=\Omega^+\cup\Gamma^+$ (in which case $\HP\calG=Y[\Omega^+]$),
otherwise $\calG=\{\perp\}$ (in which case $\HP\calG=\varnothing$).
Theorem~\ref{th:ellipsoid} gives a class of algorithms for solving [{\tt Feasibility}].
If $\HP\calG\ne\varnothing$ then $Y[\Omega^+]\ne\varnothing$ and thus $\Pi$ is not a partition of $\Gamma$, so we can output result (b) in Theorem~\ref{th:Pi-testing}.
If $\HP\calG=\varnothing$ then either $Y[\Omega^+]=\varnothing$ (in which case $\Pi$ is a partition of $\Gamma$) or $\Gamma$ is not solvable.
Thus, we can output result (a).
We obtained a correct but possibly inefficient algorithm for solving the problem in Theorem~\ref{th:Pi-testing}.
Next, we will modify this algorithm so that it remains correct and has the desired complexity.

As the first step, we compute instances $\calJ=\Feasplus{\calJ_\circ}$ and $\calJ'=\Feasplus{\calJ'_\circ}$,
and store the result.
Now suppose that the ellipsoid method calls the oracle for $[{\tt Separation}]$ with vector $y\in\mathbb Q^{\Gamma^+}$.
We implement this call as follows.
First, call the algorithm from Lemma~\ref{lemma:Ph-testing:separation}. 
If it asserts that $y\in Y[\Omega^+]$
then $Y[\Omega^+]\ne\varnothing$ and so $\Pi$ is not a partition of $\Gamma$.
Thus, we can immediately terminate the ellipsoid method 
and output result (b) in Theorem~\ref{th:Pi-testing}.
Otherwise we have a
hyperplane $c\in\mathbb Q^{\Gamma^+}\times\mathbb Q$ that satisfies either condition (i) or (ii) 
in Lemma~\ref{lemma:Ph-testing:separation}(b). We return $c$ as the output of the oracle;
clearly, in both cases it satisfies the required condition.
 It can be seen that the modified algorithm is still correct
and performs a polynomial number of operations (excluding the time for computing  $\calJ$ and $\calJ'$). 

\paragraph{Proof of Lemma~\ref{lemma:Ph-testing:separation}}
We can assume w.l.o.g.\ that the input vector $y\in\mathbb Q^{\Gamma^+}$ is non-negative
(otherwise it is easy to find a separating hyperplane~$c$). 
For a vector $z\in\mathbb Q^{\Gamma^+}_{\ge 0}$
let $\calI(z)$ and $\calI'(z)$ be the instances with variables $D$ and $\Pi$, respectively, that have the following objective functions
(we use the same convention for weights as in Section~\ref{sec:th:find-core:general:proof}):
\begin{eqnarray*}
f_{\calI(z)}(g)&=&\sum_{(f,\bxS)\in\Gamma^+} z(f,\bx)f(g_{x_1},\ldots,g_{x_n}) + f_{\calJ}(g)\qquad\quad\hspace{11pt}\forall g:D\rightarrow D \\
f_{\calI'(z)}(g)&=&\sum_{(f,\bxS)\in\Gamma^+} z(f,\bx)f(g_{[x_1]},\ldots,g_{[x_n]}) + f_{\calJ'}(g)\qquad\quad\forall g:\Pi\rightarrow D 
\end{eqnarray*}

Let $\mu$ be an optimal solution of the BLP relaxation of instance $\calI(y)$, 
assuming that a feasible solution exists. (The BLP relaxation in eq.~\eqref{eq:BLP}
also uses vector $\alpha$; however, it can be easily computed from $\mu$, so we omit it).
Note, $\mu$ can be computed in polynomial time given inputs $(\Gamma,\calJ)$ 
and $y\in\mathbb Q^{\Gamma^+}_{\ge 0}$. Vector $\mu$ has components $\mu_{f,\bxS}(\bx')\in\mathbb Q_{\ge 0}$ for $(f,\bx)\in\Gamma^+$ and $\bx'\in\dom f$,
with $\sum_{\bxS'\in\dom f}\mu_{f,\bxS}(\bx')=1$.
Define vector $c\in\mathbb Q^{\Gamma^+}$ as follows:
\begin{eqnarray*}
c(f,\bx)&=&\left[\sum_{\bxS'\in\dom f}\mu_{f,\bxS}(\bx')f(\bx')\right]-f(\bx)\qquad\qquad\forall (f,\bx)\in\Gamma^+ 
\end{eqnarray*}
Similarly, let $\mu'$ be an optimal solution of the BLP relaxation of $\calI'(y)$ (if exists), and define vector $c'\in\mathbb Q^{\Gamma^+}$ via
\begin{eqnarray*}
c'(f,\bx)&=&\left[\sum_{\bxS'\in\dom f}\mu'_{f,\bxS}(\bx')f(\bx')\right]-f(\bx)\qquad\qquad\forall (f,\bx)\in\Gamma^+ 
\end{eqnarray*}
Now do the following: \\
\indent (1) if $\mu$ exists and $\langle c,y\rangle< 0$ then return $[c\;\;0]\in\mathbb Q^{\Gamma^+}\times\mathbb Q$ as a separating hyperplane; \\
\indent (2) if $\mu'$ exists and $\langle c',y\rangle< 1$ then return $[c'\;\;-\!1]\in\mathbb Q^{\Gamma^+}\times\mathbb Q$ as a separating hyperplane; \\
\indent (3) if (1) and (2) do not apply then report that $y\in Y[\Omega^+]$. \\
If (1) and (2) are both applicable then we pick one of them arbitrarily.
This concludes the description of the algorithm in Lemma~\ref{lemma:Ph-testing:separation};
it remains to show that it is correct.

Observing that $f_{\calI+\Feas\calI}(x)=f_\calI(x)$ for any instance $\calI$ and labeling $x$, we obtain
\begin{eqnarray}\label{eq:Pi-testing-proof:fI}
f_{\calI(z)}(g)&=&
\begin{cases}
  \mathlarger\sum\limits_{(f,\bxS)\in\Gamma^+} z(f,\bx)f(g(\bx))\;<\;\infty \qquad & \mbox{if } g\in\calO^{(1)}\cap \Pol\Gamma=\Omega^+  \\
  \;\;\;\;\infty & \mbox{if } g\in\calO^{(1)} - \Pol\Gamma 
\end{cases}
\end{eqnarray}

Clearly, there is a natural isomorphism between sets $\calO_\Pi$ and $D^\Pi$.
Thus, operations $g\in\calO_\Pi$ can be equivalently viewed as labelings $g:\Pi\rightarrow D$.
With this convention, we can evaluate expression $f_{\calI'(z)}(g)$
for an operation $g\in\calO_\Pi$. It can be seen that $f_{\calI'(z)}(g)=f_{\calI(z)}(g)$ for any $g\in\calO_\Pi$, and so
\begin{eqnarray}\label{eq:Pi-testing-proof:fI'}
f_{\calI'(z)}(g)&=&
\begin{cases}
  \mathlarger\sum\limits_{(f,\bxS)\in\Gamma^+} z(f,\bx)f(g(\bx))\;<\;\infty \qquad & \mbox{if } g\in\calO_\Pi\cap \Pol\Gamma=\Omega^-  \\
  \;\;\;\;\infty & \mbox{if } g\in\calO_\Pi - \Pol\Gamma 
\end{cases}
\end{eqnarray}

Let $BLP(\calI(z),\mu)$ be the cost of solution $\mu$ in the BLP relaxation of $\calI(z)$,
and $BLP(\calI'(z),\mu')$ be the cost of solution $\mu'$ in the BLP relaxation of $\calI'(z)$.
By checking~\eqref{eq:BLP} we obtain
\begin{eqnarray*}
BLP(\calI(z),\mu) &=& \sum_{(f,\bxS)\in\Gamma^+} z(f,\bx)\sum_{\bxS'\in\dom f}\mu_{f,\bxS}(\bx')f(\bx') \\
BLP(\calI'(z),\mu') &=& \sum_{(f,\bxS)\in\Gamma^+} z(f,\bx)\sum_{\bxS'\in\dom f}\mu'_{f,\bxS}(\bx')f(\bx')
\end{eqnarray*}
Denoting
$
h(z)=\sum\nolimits_{(f,\bxS)\in\Gamma^+}
z(f,\bx)f(\bx)
$, we have
$$
\langle c,z\rangle=BLP(\calI(z),\mu)-h(z)\qquad\qquad
\langle c',z\rangle=BLP(\calI'(z),\mu')-h(z)
$$

We are now ready to prove algorithm's correctness. We need to consider three cases:

\begin{itemize}
\item[(1)]\underline{$\mu$ exists and $\langle c,y \rangle<0$.}~~
Consider vector $z\in\mathbb Q^{\Gamma^+}$ with $\langle c,z\rangle< 0$.
We will show that either $z\notin Y[\Omega^+]$ or $\Gamma$ is not solvable.
Clearly, this will imply a similar claim for vectors $z\in\mathbb R^{\Gamma^+}$.

We can assume that $z$ is non-negative and $\Gamma$ is solvable, otherwise the claim holds trivially.
By Theorem~\ref{th:KKZ:SICOMP17} BLP solves instance $\calI(z)$, therefore there exists mapping $g\in\calO^{(1)}$
such that $f_{\calI(z)}(g)\le BLP(\calI(z),\mu)$. From~\eqref{eq:Pi-testing-proof:fI} we conclude that $g\in\Omega^+$.
We can write
$$
\sum_{(f,\bxS)\in\Gamma^+} z(f,\bx)\cdot\left[f(g(\bx))-f(\bx)\right]
\;=\; f_{\calI(z)}(g)-h(z)
\;\le\; BLP(\calI(z),\mu)-h(z)
\;=\; \langle c,z\rangle
\;<\; 0
$$
Equivalently, $\langle c_g,[z\; 1]\rangle<0$. Vector $z$ violates inequality~\eqref{eq:g:inequality} for $g$, and therefore $z \notin Y[\Omega^+]$.

\item[(2)]\underline{$\mu'$ exists and $\langle c',y \rangle<1$.}~~ 
Consider vector $z\in\mathbb Q^{\Gamma^+}$ with $\langle c',z\rangle< 1$.
Similar to the previous case, we will show that either $z\notin Y[\Omega^+]$ or $\Gamma$ is not solvable.

We can assume that $z$ is non-negative and $\Gamma$ is solvable, otherwise the claim holds trivially.
By Theorem~\ref{th:KKZ:SICOMP17} BLP solves instance $\calI'(z)$, therefore there exists mapping $g\in\calO_\Pi$
such that $f_{\calI'(z)}(g)\le BLP(\calI'(z),\mu')$. From~\eqref{eq:Pi-testing-proof:fI'} we conclude that $g\in\Omega^-$.
We can write
$$
\sum_{(f,\bxS)\in\Gamma^+} z(f,\bx)\cdot\left[f(g(\bx))-f(\bx)\right]
\;=\; f_{\calI'(z)}(g)-h(z)
\;\le\; BLP(\calI'(z),\mu')-h(z)
\;=\; \langle c',z\rangle
\;<\; 1
$$
We showed that $\langle c_g,[z\; 1]\rangle<0$. This means that $z \notin Y[\Omega^+]$.

\item[(3)]\underline{Cases (1) and (2) do not hold.}~~
We need to show that $y\in Y[\Omega^+]$.
Suppose not, then there exists an element $g\in\Omega^+$ such that inequality~\eqref{eq:g:inequality} is violated for $g$.
If $g\in\Omega^+-\Omega^-$ then the BLP relaxation of $\calI(y)$ has feasible solutions (since $f_{\calI(y)}(g)<\infty$)
and we can write
$$
\langle c,y \rangle
\;=\; BLP(\calI(y),\mu)-h(y)
\;\le\;   f_{\calI(y)}(g)-h(y)
\;=\;\sum_{(f,\bxS)\in\Gamma^+} y(f,\bx)\cdot\left[f(g(\bx))-f(\bx)\right]
\;<\;0
$$
where the first inequality holds since $\mu$ is an optimal solution of the BLP relaxation of $\calI(y)$.
Thus, case (1) holds - a contradiction. If $g\in\Omega^-$ then
we obtain in a similar way that case (2) holds:
$$
\langle c',y \rangle
\;=\; BLP(\calI'(y),\mu')-h(y)
\;\le\;   f_{\calI'(y)}(g)-h(y)
\;=\;\sum_{(f,\bxS)\in\Gamma^+} y(f,\bx)\cdot\left[f(g(\bx))-f(\bx)\right]
\;<\;1
$$
\end{itemize}


\subsection{Proof of Theorem~\ref{th:24} (testing solvability of core languages)}\label{sec:th:24:proof}
Let $\calO^{(m)}_{\tt id}\subseteq \calO^{(m)}$ be the set of idempotent operations of arity $m$.
Given an input language $\Gamma$, we make the following definitions:
\begin{itemize}
\item If $\Gamma$ is finite-valued then $m=2$ and 
$
\Omega^-=\Omega^+=\{g\in\calO^{(2)}_{\tt id}\cap\Pol\Gamma\:|\:g\mbox{ is symmetric}\}
$.
\item Otherwise $m=4$,
$\Omega^+=\calO^{(4)}_{\tt id}\cap\Pol\Gamma$
and $
\Omega^-=\{g\in\Omega^+\:|\:g\mbox{ is Siggers}\}
$.
\end{itemize}
Using Lemmas~\ref{lemma:solvability:equiv} and~\ref{lemma:Farkas},
we conclude that $\Gamma$ is solvable if and only if $Y[\Omega^+]=\varnothing$.

Define undirected graph $(D^m,E)$ via $E=\{\{(x,y),(y,x)\}\:|\:x,y\in D\}$ in the case of $m=2$
or $E=\{\{(r,a,r,e),(a,r,e,a)\}\:|\:a,e,r\in D\}$ in the case of $m=4$.
Let $\Pi$ be the set of connected components of $(D^m,E)$ (it is a partition of $D^m$).
It can be seen that operation $g\in \calO^{(m)}$ is symmetric (if $m=2$) / Siggers (if $m=4$)
if and only if $g(x)=g(y)$ for all $x,y\in A\in\Pi$. In this section
the set of such operations will be denoted as $\calO_\Pi\subseteq\calO^{(m)}$.


\begin{lemma}\label{lemma:24:separation}
There exists an algorithm with complexity $2\cdot\TLPPROBE{|D|^m,\Gamma}+O(poly(\size(\Gamma)+\size(y)))$ that given a language $\Gamma$
and a vector $y\in\mathbb Q^{\Gamma^+}$ does one of the following:
\begin{itemize}
\item[(a)] Produces element $g\in\Omega^+\cup\Gamma^+$ such that $\langle c_g,[y\; 1]\rangle<0$.
\item[(b)] Asserts that one of the following holds: (i) $\Gamma$ is not solvable; (ii) $\Gamma$ is not a core. 
\end{itemize}
\end{lemma}
Before proving this lemma, let us describe how it implies Theorem~\ref{th:find-core:general}.
%
We use the ellipsoid method where for the input $\Lambda=\Gamma$ 
we define $\calG=\Omega^+\cup\Gamma^+$. 
Theorem~\ref{th:ellipsoid:plus} gives a class of  algorithms for solving [{\tt Feasibility+}].
If $\HP\calG=Y[\Omega^+]=\varnothing$ then we have a subset $\calG'\subseteq \calG$ of polynomial size with $\HP{\calG'}=0$.
Let $\Omega'=\calG'\cap\Omega^+\subseteq\Omega^+$, then $Y[\Omega']=\varnothing$.
Lemma~\ref{lemma:Farkas} gives a desired  $m$-ary fractional polymorphism $\omega$ of $\Gamma$,
which  we return as the result.
If $\HP\calG=Y[\Omega^+]\ne\varnothing$ then $\Gamma$ is not solvable 
(by the observation in the beginning of this section)
and so we can output result~(b) in Theorem~\ref{th:24}.
Next, we will modify this algorithm so that it remains correct and has the desired complexity.

Suppose that the ellipsoid method calls the oracle for $[{\tt Separation\mbox{+}}]$ with vector $y\in\mathbb Q^{\Gamma^+}$.
We implement this call as follows.
First, call the algorithm from Lemma~\ref{lemma:24:separation}. 
If it returns $g\in\Omega^+\cup\Gamma^+$ with $\langle c_g,[y\; 1]\rangle<0$ then we return $g$ as the output of the oracle.
Otherwise it asserts that $\Gamma$ is not solvable or not a core;
we can then  terminate the ellipsoid method
and output result (b) in Theorem~\ref{th:24}.

\paragraph{Proof of Lemma~\ref{lemma:24:separation}}
We can assume that the input vector $y\in\mathbb Q^{\Gamma^+}$ is non-negative
(otherwise we can easily find element $(f,\bx)\in\Gamma^+$ with $\langle c_{(f,\bxS)},[y \;1]\rangle<0$). 
Let $\calI(y)$ and $\calI'(y)$ be the instances with variables $D^m$ and $\Pi$, respectively, that have the following objective functions
(we use the same convention for weights as in Section~\ref{sec:th:find-core:general:proof}):
\begin{eqnarray*}
f_{\calI(y)}(g)&=&\sum_{(f,\bxS)\in\Gamma^+} y(f,\bx)f(g_{x_1},\ldots,g_{x_n}) + \sum_{a\in D}u_a(g_{(a,\ldots,a)})\qquad\quad\hspace{11.5pt}\forall g:D^m\rightarrow D \\
f_{\calI'(y)}(g)&=&\sum_{(f,\bxS)\in\Gamma^+} y(f,\bx)f(g_{[x_1]},\ldots,g_{[x_n]}) + \sum_{a\in D}u_a(g_{[(a,\ldots,a)]}) \!\qquad\quad\forall g:\Pi\rightarrow D 
\end{eqnarray*}
Note that $\calI(y),\calI'(y)\in\VCSP{\Gamma\cup\{u_a\:|\:a\in D\}}$. Similar to previous sections, we have
\begin{eqnarray}\label{eq:24:A}
f_{\calI(y)}(g)&=&
\begin{cases}
  \mathlarger\sum\limits_{(f,\bxS)\in\Gamma^+} y(f,\bx)f(g(\bx))\;<\;\infty \qquad & \mbox{if } g\in\calO^{(m)}_{\tt id}\cap \Pol\Gamma  \\
  \;\;\;\;\infty & \mbox{if } g\in\calO^{(m)} - (\calO^{(m)}_{\tt id}\cap \Pol\Gamma)
\end{cases}
\\
f_{\calI'(y)}(g)&=&
\begin{cases}
  \mathlarger\sum\limits_{(f,\bxS)\in\Gamma^+} y(f,\bx)f(g(\bx))\;<\;\infty \qquad & \mbox{if } g\in\calO_\Pi\cap\calO^{(m)}_{\tt id}\cap \Pol\Gamma=\Omega^-  \\
  \;\;\;\;\infty & \mbox{if } g\in\calO_\Pi - (\calO^{(m)}_{\tt id}\cap\Pol\Gamma)
\end{cases}
\label{eq:GHLASJHAG}
\end{eqnarray}
where we used the natural isomorphism between sets $\calO_\Pi$ and $D^\Pi$
to evaluate expression $f_{\calI'(y)}(g)$ for $g\in\calO_\Pi$. 
Let us denote
$
h(y)=\sum\limits_{(f,\bxS)\in\Gamma^+}
y(f,\bx)f^m(\bx)
$.
We now run the following steps.
\begin{enumerate}
\item Compute $g'=\mbox{\tt LP-Probe}(\calI'(y),D)$. If $g'={\tt FAIL}$ then $\Gamma$ is either not solvable or not a core (by Lemmas~\ref{lemma:solvability:equiv} and~\ref{lemma:LPProbe}), so we output result (b)
in Lemma~\ref{lemma:24:separation} and terminate. 
If $g'=\varnothing$, then $\calI'(y)$ is infeasible and thus $\Omega^-=\varnothing$
by eq.~\eqref{eq:GHLASJHAG}. This means that $\Gamma$ is either not solvable or not a core, and so we can
again output result
(b)
in Lemma~\ref{lemma:24:separation} and terminate. 
If $g'$ is a labeling with $f_{\calI'(y)}(g')<h(y)+1$
then $g'\in\Omega^-$ and $\langle c_{g'},[y\; 1] \rangle=f_{\calI'(y)}(g')-h(y)-1<0$,
so we can output element $g'$ in Lemma~\ref{lemma:24:separation} and terminate. 

From now on
we assume that $g'$ is a labeling satisfying $f_{\calI'(y)}(g')\ge h(y)+1$.
This means that for any $\hat g\in\Omega^-$ we have 
$$
\langle c_{\hat g},[y\; 1] \rangle
\;\;=\;\;f_{\calI'(y)}(\hat g)-h(y)-1
\;\;\ge\;\; f_{\calI'(y)}(g')-h(y)-1
\;\;\ge\;\; 0
$$
\item {(This step is run only if $m=4$)} Compute $g=\mbox{\tt LP-Probe}(\calI(y),D)$. 
If $g={\tt FAIL}$ then we output result (b)
in Lemma~\ref{lemma:24:separation} and terminate (by the same argument as above). 
If $g=\varnothing$, then $\calI(y)$ is infeasible and thus $\calO^{(m)}_{\tt id}\cap \Pol\Gamma=\varnothing$
by eq.~\eqref{eq:24:A}. This means that $\Gamma$ is not a core, and so we can
again output result
(b)
in Lemma~\ref{lemma:24:separation} and terminate. 
If $g$ is a labeling with $f_{\calI(y)}(g)<h(y)$
then $g\in\calO^{(4)}_{\tt id}\cap\Pol\Gamma=\Omega^+$ and $\langle c_{g},[y\; 1] \rangle=f_{\calI(y)}(g)-h(y)-[g\in\Omega^-]<0$,
so we can output element $g$ in Lemma~\ref{lemma:24:separation} and terminate. 

From now on 
we assume that $g$ is a labeling satisfying $f_{\calI(y)}(g)\ge h(y)$.
This means that for any $\hat g\in\Omega^+-\Omega^-$ we have 
$$
\langle c_{\hat g},[y\; 1] \rangle
\;\;=\;\;f_{\calI(y)}(\hat g)-h(y)
\;\;\ge\;\; f_{\calI(y)}(g)-h(y)
\;\;\ge\;\; 0
$$
\item If we reached this point, we know that $\langle c_{\hat g},[y\; 1] \rangle\ge 0$ for all $\hat g\in\Omega^+$,
and therefore $y\in Y[\Omega^+]$. 
Thus, $Y[\Omega^+]\ne \varnothing$ and so $\Gamma$ is not solvable (by the observation in the beginning of this section),
therefore we can output result (b) in Lemma~\ref{lemma:24:separation}.
\end{enumerate}

\subsection{Proof of Theorems~\ref{th:ETH-hardness} and~\ref{th:SETH-hardness} (hardness under ETH/SETH)}
Let us fix constant $L\in\{1,\infty\}$, and denote $\calL=\{0,L\}$.
For a function $f:D^n\rightarrow\{0,\infty\}$ over some domain $D$
let  $f^\calL:D^n\rightarrow\calL$  be the function with $f^\calL(x)=\min\{f(x),L\}$.
In particular, $u_A^\calL$ for a subset $A\subseteq D$ is the function $u_A^\calL:D\rightarrow\calL$ with $\argmin u_A^\calL=A$.

The proof will be based on the following construction.
Consider a CSP instance $\calI$ with variables $V$, domain $[c]$ and the following objective function:
\begin{equation}
f_{\calI}(x)\ =\ \sum_{t\in T} f_t(x_{v(t,1)},\ldots,x_{v(t,n_t)})\qquad \forall x:V\rightarrow [c]
\end{equation}
where all terms $f_t$ are $\{0,\infty\}$-valued functions satisfying $\min_{x\in[c]^{n_t}} f_t(x)=0$.
To such  $\calI$ we associate the following language $\Gamma=\Gamma^\calL(\calI)$ on domain $D=V\times[c]$:
$$
\Gamma=\left\{f_t^{\langle v(t,1),\ldots,v(t,n_t)\rangle} \: | \: t\in T\right\}
\cup \left\{u^\calL_{D_v}\:|\: v\in V\right\}
$$
where for an $n$-ary function $f$ over $[c]$ and variables $v_1,\ldots,v_n\in V$,
function $f^{\langle v_1,\ldots,v_n\rangle}:D^n\rightarrow \calL$ is defined via
$$
f^{\langle v_1,\ldots,v_n\rangle}(y)=\begin{cases}
f^\calL(x) & \mbox{if }y=((v_1,x_1),\ldots,(v_n,x_n))  \\
L & \mbox{otherwise}
\end{cases}\qquad\forall y\in D^n
$$
Note, this construction is similar to the technique of ``lifting a language'' introduced in~\cite{Kolmogorov2015:hybrid};
it was also used in~\cite{KKZ:SICOMP17} in a different context.

\begin{theorem}[\cite{ChenLarose:17}]\label{th:GAKGJHAKSGJ}
Suppose that instance $\calI$ expresses a 3-coloring problem in an undirected graph $(V,E)$, 
i.e.\ $c=3$ and $\calI$ has the following objective function for labeling $x:V\rightarrow \{1,2,3\}$:
\begin{equation}
f_{\calI}(x)\ =\ \sum_{\{u,v\}\in E} F(x_u,x_v)\qquad\mbox{where}\qquad
F(a,b)=\begin{cases}
0 & \mbox{if }a\ne b\\
\infty & \mbox{if }a=b
\end{cases}
\end{equation}
Then language $\Gamma=\Gamma^\calL(\calI)$ has the following properties:\\
(a) If $\min_{x} f_{\calI}(x)=0$ then $\coresize(\Gamma)=|V|$ and $\Gamma$ is solvable. \\
(b) If $\min_{x} f_{\calI}(x)\ne 0$ then $\coresize(\Gamma)=3\cdot |V|$ (i.e.\ $\Gamma$ is a core) and $\Gamma$ is not solvable.
\end{theorem}
\begin{remark}
Chen and Larose~\cite{ChenLarose:17} showed this statement for the case $\calL=\{0,\infty\}$
and for the language $\tilde\Gamma=\Gamma-\left\{u^\calL_{D_v}\:|\: v\in V\right\}$ without unary terms. However, the proof 
can easily be extended to our setting. Indeed, part (a) is an ``easy'' direction (see e.g.\ Lemma~\ref{lemma:GAKSHAKSJFG} below).
Part~(b) holds for language $\tilde\Gamma^{\{0,\infty\}}(\calI)$ \cite{ChenLarose:17}
and thus for language $\Gamma^{\{0,\infty\}}(\calI)$.
Observing that $\fPolplus{\Gamma^{\{0,1\}}(\calI)}\subseteq \fPolplus{\Gamma^{\{0,\infty\}}(\calI)}=\Pol{\Gamma^{\{0,\infty\}}(\calI)}$
yields the claim for language $\Gamma^{\{0,1\}}(\calI)$.
\end{remark}
\begin{theorem}[{\cite[Theorem 3.2]{lokshtanov11:survey}}]\label{th:AKJSGBKAS}
Suppose that ETH holds. Then the 3-coloring problem for a graph on $n$ nodes cannot be solved in $2^{o(n)}$ time.
\end{theorem}
Clearly, Theorem~\ref{th:ETH-hardness} follows immediately from Theorems~\ref{th:GAKGJHAKSGJ} and~\ref{th:AKJSGBKAS}.
(Note that in the construction above we have $\size(\Gamma^\calL(\calI))=O(poly(\size(\calI)))$, $\size(\calI)=O(poly(|V|))$ and $d=3\cdot |V|$).
We now turn to the proof of Theorem~\ref{th:SETH-hardness}. First, we establish some connections between instance~$\calI$
and language $\Gamma^\calL(\calI)$.
\begin{lemma}\label{lemma:GAKSHAKSJFG}
Language $\Gamma=\Gamma^\calL(\calI)$ has the following properties: \\
(a)  $\coresize(\Gamma)\ge |V|$. Furthermore, $\coresize(\Gamma)=|V|$ if and only if $\min_{x} f_{\calI}(x)=0$. \\
(b) If $\min_{x} f_{\calI}(x)=0$ then $\Gamma$ is solvable. \\
(c) There exists an algorithm for deciding whether $\min_{x} f_{\calI}(x)=0$ that works as follows:
\begin{itemize}
\item If $L=1$ then it makes one oracle call to test solvability of $\Gamma$ and performs a polynomial number of other operations.
\item If $L=\infty$ then it makes $|D|+1$ oracle calls to test solvability of languages of the form $\Gamma\oplus B\eqdef \Gamma\cup\{u_a\:|\:a\in B\}$
for subsets $B\subseteq D$ and performs a polynomial number of other operations, assuming
the existence of a uniform polynomial-time algorithm for core crisp languages.
\end{itemize}
\end{lemma}
Before giving a proof of this lemma, let us show how it implies Theorem~\ref{th:SETH-hardness}.
Let $\delta$ be the constant from Theorem~\ref{th:SETH-hardness}. Fix constant $\delta',\delta''$
and positive integers $p,q$ such that
$\delta<\delta'<\delta''<1$
and  $\frac{\delta}{\delta'}\log_2 3\le \frac{p}{q}\le\log_2 3$.
Assume that SETH holds, then there exists $k$ such that satisfiability of a $k$-CNF-SAT formula 
 $\varphi=C_1\wedge\ldots \wedge C_m$ with $n$ Boolean variables and $m$ clauses
cannot be decided in $O(2^{\delta''n})$ time. Note that $m=O(poly(n))$, since $k$ is a constant.
Next, we construct CSP instance $\calI=\calI(\varphi)$ 
such that $\varphi$ is satisfiable if and only if $\calI$ has a feasible solution.
The idea is to encode labelings in $\{0,1\}^p$ via labelings in $\{1,2,3\}^q$.
Specifically, let us fix an injective mapping $\sigma:\{0,1\}^p\rightarrow \{1,2,3\}^q$
(it exists since $2^p\le 3^q$). 
Denote $r=\lceil\frac{n}{p}\rceil$.
We can treat $\varphi$ as a formula over $rp$ Boolean variables
(by adding dummy variables, if necessary).
Thus, we can name variables of $\varphi$ as follows: $x=(x^1,\ldots,x^r)$
where $x^i=(x^i_1,\ldots,x^i_p)\in\{0,1\}^p$ for each $i\in[r]$.
Instance $\calI$ will have $rq$ variables $y\in \{1,2,3\}^{rq}$.
It will be convenient to write $y=(y^1,\ldots,y^r)$ where $y^i=(y^i_1,\ldots,y^i_q)\in \{1,2,3\}^q$
for each $i\in[r]$. To construct the objective function of $\calI$, we do the following for each clause $C_\ell$ of $\varphi$:
\begin{itemize}
\item Clause $C_\ell$ can be viewed as a function $f_\ell(x^{i_1},\ldots,x^{i_{k'}})\in\{0,\infty\}$ of $k'p$ variables
(some of which may be ``dummy''), where $k'\le k$. Add the following term to $\calI$:
$$
f'_\ell(y^{i_1},\ldots,y^{i_{k'}})=\begin{cases}
f_\ell(\sigma^{-1}(y^{i_1}),\ldots,\sigma^{-1}(y^{i_{k'}})) & \mbox{if }y^{i_1},\ldots,y^{i_{k'}}\in \sigma(\{0,1\}^p)\subseteq \{1,2,3\}^q \\
\infty & \mbox{otherwise}
\end{cases}
$$
\end{itemize}
By construction, formula $\varphi$ is satisfiable if and only if instance $\calI$ has a feasible solution.

Let $\calF$ be the family of languages of the form $\Gamma=\Gamma^\calL(\calI(\varphi))$ for $k$-CNF-SAT formulas $\varphi$.
Clearly, family $\calF$ is $O(1)$-bounded. (Note that we have $d=3rq=3q\lceil\frac{n}{p}\rceil=\Theta(n)$ and $\size(\calI)=O(poly(m+n))=O(poly(n))$).
We claim that Theorem~\ref{th:SETH-hardness} holds for family $\calF$.
Indeed, if there exists an algorithm with complexity $O(\sqrt[3]3^{\,\delta d})$
for either problem (a) or problem (b) then
we can use the algorithm in Lemma~\ref{lemma:GAKSHAKSJFG}(c)
to test feasibility of instance $\calI(\varphi)$ 
(and thus satisfiability of~$\varphi$) in time
 $O(\sqrt[3]3^{\,\delta d}\cdot poly(d))$.
We have 
$
\sqrt[3]3^{\,\delta d}
=3^{\delta q\lceil\frac{n}{p}\rceil}
=3^{\delta q\frac{n}{p}+o(1)}
\le 2^{\delta' n+o(1)}
$, so we obtain an algorithm for $k$-CNF-SAT with complexity $O(2^{\delta'' n})$ - a contradiction.

\paragraph{Proof of Lemma~\ref{lemma:GAKSHAKSJFG}}
For any node $v\in V$ and core $B\in\Bcore\Gamma$ we have $|B\cap D_v|\ge 1$
(since the presence of function~$u^\calL_{D_v}$ in $\Gamma$ implies that $g(D_v)\subseteq D_v$ for any $g\in\fPolplus\Gamma$).
Therefore, we have $\coresize(\Gamma)=|B|\ge |V|$.
We make the following claims:
\begin{itemize}
\item {\em If $\min_x f_\calI(x)=0$ then $\coresize(\Gamma)=|V|$ and $\Gamma$ is solvable.}
Indeed, pick $x^\ast\in [c]^V$ with $f_\calI(x^\ast)=0$, and let $g:D\rightarrow D$ be the operation
with $g((v,a))=x^\ast_v$ for all $a\in [c]$. It can be checked that vector $\chi_g$ is a fractional
polymorphism of $\Gamma$. (Note, if $f_t^{\langle v_1,\ldots,v_n\rangle}(y)=0$
then $y=((v_1,x_1),\ldots,(v_n,x_n))$ and so $g(y)=((v_1,x^\ast_1),\ldots,(v_n,x^\ast_n))$
and $f_t^{\langle v_1,\ldots,v_n\rangle}(g(y))=0$).
This implies that set $B=\{(v,x^\ast_v)\:|\:v\in V\}$ is a core of $\Gamma$ and $\coresize(\Gamma)=V$.

To show that $\Gamma$ solvable, it suffices to show that $\Gamma'=\Gamma[B]$ is solvable (by Lemma~\ref{lemma:solvability:equiv}).
Language $\Gamma'$ contains unary functions $\{u^\calL_a\:|\:a\in B\}$ and functions of the form
$$
f^{\langle v_1,\ldots,v_n\rangle}(y)=\begin{cases}
0 & \mbox{if }y=((v_1,x^\ast_1),\ldots,(v_n,x^\ast_n))  \\
L & \mbox{otherwise}
\end{cases}\qquad\forall x\in B^n
$$
It can be checked that vector $\frac{1}{3}(\chi_{g_1}+\chi_{g_2}+\chi_{g_3})$
is a fractional polymorphism of $\Gamma'$, where $(g_1,g_2,g_3):B^3\rightarrow B^3$ are ternary operations
defined via 
$$
(g_1,g_2,g_3)(a,b,c)=\begin{cases}

({\tt mjrt}(a,b,c),{\tt mjrt}(a,b,c),{\tt mnrt}(a,b,c)) & \mbox{if } |\{a,b,c\}|\le 2 \\
(a,b,c) & \mbox{otherwise}
\end{cases}
$$
and functions ${\tt mjrt},{\tt mnrt}$ return respectively the majority and the minority element among its arguments.
We obtain that ${\tt mjrt}\in\fPolplus{\Gamma'}$.
It is well-known that such $\Gamma'$ is solvable.\footnote{One can use, for example, Lemma~\ref{th:clone}
and the fact that a clone of operations containing a near-unanimity operation (such as ${\tt mjrt}$)
also contains a Siggers operation, see~\cite{CSPsurvey:17}.
}

\item {\em If $\coresize(\Gamma)=|V|$ then $\min_x f_\calI(x)=0$.}
Indeed, pick a core $B\in\Bcore\Gamma$. Since $|B\cap D_v|\ge 1$ for each $v\in V$, we have $|B\cap D_v|= 1$ for each $v\in V$.
Define labeling $x^\ast\in[c]^V$
so that $x^\ast_v$ is the unique label in $B\cap D_v$. Pick vector $\omega\in\fPol[1]\Gamma$ such that $g(D)=B$ for all
$g\in\supp(\omega)$ (it exists by Lemma~\ref{lemma:core-basics}(b)).
Consider $g\in\supp(\omega)$. We must have $g(D_v)\subseteq D_v$ due to the presence of function $u^\calL_{D_v}$ in $\Gamma$.
Thus, $g(D_v)=\{x^\ast_v\}$ and so $\supp(\omega)=\{g\}$.
We showed that $\chi_g\in\fPol[1]\Gamma$. Now consider function $f_t^{\langle v_1,\ldots,v_n \rangle}\in\Gamma$.
By assumption, there exists tuple $y$ with $f_t^{\langle v_1,\ldots,v_n \rangle}(y)=0$.
We can write
$$
f_t^{\langle v_1,\ldots,v_n \rangle}((v_1,x^\ast_{v_1}),\ldots,(v_n,x^\ast_{v_n}))=f_t^{\langle v_1,\ldots,v_n \rangle}(g(y))\le f_t^{\langle v_1,\ldots,v_n \rangle}(y)=0
$$
Thus, $f_t(x^\ast_{v_1},\ldots,x^\ast_{v_n})=0$.
This means that $f_\calI(x^\ast)=0$, as claimed.
\end{itemize}
We now describe the algorithm for part (c). First, we test whether language $\Gamma$ is solvable.
Assume that it is (otherwise we can report that $\min_x f_\calI(x)\ne 0$ and terminate).
Let $\calJ$ be the instance with variables $V$ and domain $D$ obtained
from $\calI$ by replacing each term $f_t$ with $f^{\langle v(t,1),\ldots,v(t,n_t)\rangle}_t$:
\begin{equation}
f_{\calJ}(y)\ =\ \sum_{t\in T} f^{\langle v(t,1),\ldots,v(t,n_t)\rangle}_t(y_{v(t,1)},\ldots,y_{v(t,n_t)})\qquad \forall y:V\rightarrow D
\end{equation}
It can be verified that $\min_x f_{\calI}(x)=0$ if and only if $\min_y f_{\calJ}(y)=0$. Also, $\calJ$ is a $\Gamma$-instance with
$\size(\calJ)=O(poly(\size(\calI))$. If $L=1$ then we have $\min_y f_{\calJ}(y)=BLP(\calJ)$ by Theorem~\ref{thm:BLP},
so this minimum can be computed in polynomial time via linear programming. Now suppose that $L=\infty$.
We initialize $B=\varnothing$ and then go through labels $b\in D$ in some
order and do the following:
\begin{itemize}
\item Test whether language $\Gamma\oplus B\oplus\{b\}$ is solvable. If it is then add $b$ to $B$.
\end{itemize}
Let $B\subseteq D$ be the set upon termination.
By construction, $\Gamma\oplus B$ is solvable. We claim that $B$ is a core of $\Gamma\oplus B$.
Indeed, consider core $B'\in\Bcore{\Gamma\oplus B}$. Clearly, we have $B\subseteq B'$.
Suppose that there exists $b\in B'-B$. 
By Lemma~\ref{lemma:solvability:equiv}(d), language $\Gamma\oplus B\oplus\{b\}$ is solvable.
By construction, there exists subset $\tilde B\subseteq B$ such
that language $\Gamma\oplus\tilde B\oplus\{b\}$ is not solvable.
We obtained a contradiction.

Since $B$ is a core of $\Gamma\oplus B$, we have $\min_{y\in D^V} f_\calJ(y)=\min_{y\in B^V} f_\calJ(y)$.
The latter minimum can be obtained by calling the assumed uniform polynomial-time algorithm for a core crisp language
$(\Gamma\oplus B)[B]$.

\appendix

\section{Proof of Lemmas~\ref{lemma:core-basics} and~\ref{lemma:Pi} (properties of cores)}\label{sec:core:proofs}
In this section we fix language $\Gamma$ on a domain $D$, and denote $\mathbb G=\fPolplus[1]\Gamma$.
It is well-known that set $\mathbb G$ is closed under compositions. 
Indeed, for vectors $\omega,\omega'\in\fPol[1]\Gamma$ we can write
\begin{equation}\label{eq:ALKSGHAKSF}
f(x)
\ge \sum_{g\in\supp(\omega)}\omega(g)f(g(x))
\ge \sum_{g\in\supp(\omega)}\omega(g)\sum_{h\in\supp(\omega')}\omega'(h)f(h(g(x))
\quad\forall f\in\Gamma,x\in\dom f
\end{equation}
This means that vector $\nu=\sum_{g\in\supp(\omega),h\in\supp(\omega')}\omega(g)\omega'(h)\chi_{h\circ g}$ is a unary fractional polymorphism of $\Gamma$
with $\supp(\nu)=\{h\circ g\:|\:g\in\supp(\omega),h\in\supp(\omega')\}$. The claim follows.

We will use a tool from~\cite{kolmogorov15:power} that allows to construct
new fractional polymorphisms from existing ones using superpositions. The following claim
is a special case of the ``Expansion Lemma'' in~\cite{kolmogorov15:power} for unary fractional polymorphisms.\footnote{
The Expansion Lemma uses the notion of an ``expansion operator'' $\Exp$ that assigns to each $g\in\mathbb G$ a probability
distribution $\omega=\Exp(g)$ over $\mathbb G$. In the case of Lemma~\ref{lemma:expansion:1} this
distribution would be defined as $\omega=\sum_{h\in\supp(\nu)} \nu(h)\chi_{h\circ g}$, where $\nu$ is a fractional polymorphism of $\Gamma$
with $\supp(\nu)=\mathbb G$. (Such $\nu$ exists due to the following observation:
if $\nu',\nu''\in\fPol[1]\Gamma$ then $\frac{\nu'+\nu''}{2}\in\fPol[1]\Gamma$ and $\supp(\frac{\nu'+\nu''}{2})=\supp(\nu')\cup\sup(\nu'')$).
 It can be seen
that this operator $\Exp$ is ``valid'' in the terminology of~\cite{kolmogorov15:power}, i.e.\ every $f\in\Gamma$ and $x\in\dom f$ satisfies
$
\sum_{h\in \supp(\omega)}\omega(h)f(h(x))\le f(g(x))
$.
}
\begin{lemma}[\cite{kolmogorov15:power}] \label{lemma:expansion:1}
Let $\mathbb G^\ast$ be a subset of  $\mathbb G=\fPolplus[1]\Gamma$ 
such that for any $g\in \mathbb G$ 
there exists a sequence $h_1,\ldots,h_r\in\mathbb G$, $r\ge 0$ with $h_r\circ \ldots \circ h_1 \circ g \in\mathbb G^\ast$. 
Then there exists a unary fractional polymorphism $\omega$ of $\Gamma$ with $\supp(\omega)\subseteq\mathbb G^\ast$.
\end{lemma}
We are now ready to prove Lemmas~\ref{lemma:core-basics} and~\ref{lemma:Pi}.

\paragraph{Proof of Lemma~\ref{lemma:core-basics}(a)}
Consider operation $g\in\mathbb G$, set $B=g(D)$ and language $\Gamma'=\Gamma[B]$. 
Using the same argument as in~\eqref{eq:ALKSGHAKSF}, we conclude the following:
\begin{itemize}
\item[(i)] If $h'\in\fPolplus{\Gamma'}$ then $h'\circ g\in\fPolplus\Gamma$.
\item[(ii)] If $h\in\fPolplus\Gamma$ then $g\circ h\in\fPolplus{\Gamma'}$.
\end{itemize}

Suppose that $\Gamma'$ is not a core. Then $\fPolplus{\Gamma'}$ contains operation $h'$
with $|h'(B)|<|B|$. Operation $h=h'\circ g\in\fPolplus\Gamma$ satisfies $|h(D)|=|h'(B)|<|B|$,
and so $|B|>\coresize(\Gamma)$.

Conversely, suppose that $|B|>\coresize(\Gamma)$.
Then there exists operation $h\in\fPolplus[1]\Gamma$ with $|h(D)|<|B|$.
Operation $h'=g\circ h\in\fPolplus{\Gamma'}$ satisfies $|h'(B)|\le |h(D)|<|B|$,
and so $\Gamma'$ is not a core.

\paragraph{Proof of Lemma~\ref{lemma:core-basics}(b)}
Suppose that $\nu\in\fPol[1]\Gamma$, $h\in\supp(\nu)$ and $B=h(D)$.
Let us apply Lemma~\ref{lemma:expansion:1}
with $\mathbb G^\ast=\{g\in\mathbb G\:|\:g(D)\subseteq B\}$.
The lemma's precondition clearly holds (note, $h\circ g\in\mathbb G^\ast$ for any $g\in \mathbb G$).
This shows the existence of $\omega\in\fPol[1]\Gamma$ with $\supp(\omega)\subseteq\mathbb G^\ast$.

Now suppose in addition that $B$ is a core. We again use Lemma~\ref{lemma:expansion:1},
but now with the set $\mathbb G^\ast=\{h\in\mathbb G\:|\:g(D)=B,g(a)=a\;\;\forall a\in B\}$.
Let us show that the lemma's precondition holds for $g\in\mathbb G$.
We have $(h\circ g)(D)=B$, so mapping $h\circ g$ acts as a permutation on $B$.
Thus, there exists $k\ge 1$ such that $(h\circ g)^k$ acts as the identity mapping on $B$,
implying that $(h\circ g)^k\in\mathbb G^\ast$. This proves the claim.

\paragraph{Proof of Lemma~\ref{lemma:core-basics}(c)}
Consider $x^\ast\in\argmin_{x\in D^V}f_\calI(x)$, and suppose that $f_\calI(x^\ast)<\infty$. Clearly, any vector $\omega\in\fPol[1]\Gamma$
is a fractional polymorphism of $f_\calI$,
so $$
\sum_{g\in\supp(\omega)}\omega(g)f_\calI(g(x^\ast)) \le f_\calI(x^\ast)
$$
Thus, $g(x^\ast)\in\argmin_{x\in D^V}$ for any $g\in\supp(\omega)$
(and in fact for any $g\in\mathbb G$). The claim follows.

\paragraph{Proof of Lemma~\ref{lemma:Pi}(a)}
For an operation $h\in\mathbb G$ let $\Pi_h$  be the partition of $D$ induced
by the following equivalence relation $\sim_h$ on $D$: $a\sim_h b$ if $h(a)=h(b)$.
We have $h\in\calO_{\Pi_h}$ and thus $\Pi_h$ is a partition of $\Gamma$.
Now let $\Pi$ be a maximal partition of $\Gamma$, and fix operation $h\in\calO_\Pi\cap\mathbb G$.
We must have $h(a)\ne h(b)$ if $a,b$ belong to different components of $\Pi$
(otherwise $\Pi_h$ would be a coarser partition of $\Gamma$ than $\Pi$, a contradiction).
Thus, $|h(D)|=|\Pi|$. Also,
the following holds for~any~$g\in\mathbb G$:
$g(a)\ne g(b)$ for any distinct $a,b\in h(D)$ (otherwise $\Pi_{g\circ h}$ would be a coarser partition of $\Gamma$ than~$\Pi$). Consequently, $|g(D)|\ge |h(D)|$.
We proved that $|\Pi|=|h(D)|=\coresize(\Gamma)$.


Now fix operation $g\in\mathbb G$, and denote $B=g(D)$. We will
show that $B$ is a core of $\Gamma$ if and only if $B\in\BcorePi$.
Clearly, this will imply the lemma's claims ($\Bcore\Gamma\subseteq\BcorePi$ and $\Ocore\Gamma=\OcorePi\cap\mathbb G$).

If $B\in\BcorePi$ then $|g(D)|=|B|=|\Pi|=\coresize(\Gamma)$. Thus, $B\in\Bcore\Gamma$ by Lemma~\ref{lemma:core-basics}(a).

Conversely, suppose that $B\in\Bcore\Gamma$. 
By Lemma~\ref{lemma:core-basics}(a)  $|B|=\coresize(\Gamma)$,
and so $|B|=|\Pi|$. If there exists component $A\in\Pi$
with   $|B\cap A|\ge 2$
then $|(h\circ g)(D)|=|h(B)|<|B|=\coresize(\Gamma)$, a contradiction.
We showed that $|B\cap A|\le 1$ for all $A\in\Pi$. Thus, we must have $|B\cap A|=1$ for all $A\in\Pi$, i.e.\ $B\in\BcorePi$.

\paragraph{Proof of Lemma~\ref{lemma:Pi}(b)}
It can be seen that $|\BcorePi|=\prod_{A\in\Pi}|A|$.
The lemma's claim is thus equivalent to the following statement: $c(d_1,\ldots,d_k)\le \alpha^{d_1+\ldots+d_k}$
for positive integers $d_1,\ldots,d_k$ with $k\ge 1$,
where we defined $c(d_1,\ldots,d_k)=d_1\cdot \ldots\cdot d_k$ and denoted  $\alpha=\sqrt[3]3$.

To prove this, we use induction on $k$.
It can be checked that $d\le \alpha^d$ for any integer $d\ge 1$, so the base case $k=1$ holds.
For the induction step (with $k\ge 2$) we can write
 $c(d_1,d_2,\ldots,d_k)=c(d_1)\cdot c(d_2,\ldots,d_k)\le \alpha^{d_1}\cdot\alpha^{d_2+\ldots+d_k}=\alpha^{d_1+\ldots+d_k}$,
where the inequality holds by the induction hypothesis.

\section{Proof of Lemma~\ref{lemma:solvability:equiv} (characterizations of solvability)}\label{sec:lemma:solvability:equiv:proof}

A set $\calC \subseteq \calO_D$ is a {\em clone of operations} (or simply a {\em clone}) if it contains all
projections on $D$ (i.e.\ operations $e^{(m)}_k\in \calO^{(m)}$
with $e^{(m)}_k(x_1,\ldots,x_m)=x_k$) and is closed under superposition. It can be seen that Lemma~\ref{lemma:solvability:equiv}
follows from the four results below (and from Theorem~\ref{thm:BLP}(b)).

\begin{theorem}[\cite{Kozik15:algebraic}] \label{th:clone}
Set $\fPolplus\Gamma$ is a clone.
\end{theorem}
\begin{theorem}[\cite{Siggers:1,Siggers:2}] \label{th:clone-Siggers}
Let $\calC$ be a clone of operations. If $\calC$ contains a cyclic operation of some arity $m\ge 2$
then it also contains a Siggers operation.
\end{theorem}

\begin{theorem}[\cite{Kozik15:algebraic}]
If a core language $\Gamma$ admits a cyclic fractional polymorphism of some arity~$m$
then language $\Gamma\cup\{u_a\:|\:a\in B\}$ also admits a cyclic fractional polymorphism of the same arity.
\end{theorem}
\begin{lemma}
Consider language $\Gamma$ on domain $D$ and language $\Gamma'=\Gamma[B]$,
where $B=g(D)$ for some $g\in\fPolplus[1]\Gamma$.
Then $\Gamma$ admits a cyclic fractional polymorphism of arity $m$
if and only if $\Gamma'$ admits a cyclic fractional polymorphism of arity $m$.
\end{lemma}
\begin{proof}
Let $\nu$ be a unary fractional polymorphism of $\Gamma$
such that $g(D)\subseteq B$ for all $g\in\supp(\omega)$ (it exists by Lemma~\ref{lemma:core-basics}).
We now consider the two directions; below $f$ is  a function in $\Gamma$ of arity $n$.
\begin{enumerate}
\item
Let $\omega$ be an $m$-ary fractional polymorphism of $\Gamma$.
For any $x^1,\ldots,x^m\in(\dom f)\cap B^n$ we have
\begin{eqnarray*}
\frac{1}{m}\sum_{i=1}^m f(x^i)
&\ge& \sum_{g\in\supp(\omega)}  \omega(g)f(g(x^1,\ldots,x^m)) \\
&\ge& \sum_{g\in\supp(\omega)}  \omega(g)\sum_{h\in\supp(\nu)}\nu(h)f(h(g(x^1,\ldots,x^m)))
\end{eqnarray*}
Thus, vector $\omega'=\sum_{g\in\supp(\omega),h\in\supp(\nu)}\omega(g)\nu(h)\chi_{h\circ g}$ is an $m$-ary fractional polymorphism of $\Gamma'$
(where we have $g:B^m\rightarrow D$, $h:D\rightarrow B$ and $h\circ g:B^m\rightarrow B$).
It can be checked that if $g$ satisfies some identity such as $g(x_1,x_2,\ldots,x_m)=g(x_2,\ldots,x_m,x_1)$ then
$h\circ g$ also satisfies the same identity. The claim follows.

\item
Let $\omega'$ be an $m$-ary fractional polymorphism of $\Gamma'$.
For any $x^1,\ldots,x^m\in\dom f$ we have
\begin{eqnarray*}
\frac{1}{m}\sum_{i=1}^m f(x^i)
&\ge& \frac{1}{m}\sum_{i=1}^m \sum_{h\in\supp(\nu)}\nu(h) f(h(x^i)) \\
&\ge& \sum_{h\in\supp(\nu)}\nu(h) \sum_{g\in\supp(\omega')}  \omega'(g)f(g(h(x^1),\ldots,h(x^m)))
\end{eqnarray*}
Thus, vector $\omega=\sum_{g\in\supp(\omega'),h\in\supp(\nu)}\omega'(g)\nu(h)\chi_{g\circ [h,\ldots,h]}$
is an $m$-ary fractional polymorphism of $\Gamma$
(where we have $[h,\ldots,h]:D^m\rightarrow B^m$, $g:B^m\rightarrow B$ and $g\circ [h,\ldots,h]:D^m\rightarrow B$).
It can be checked that if $g$ satisfies some identity such as $g(x_1,x_2,\ldots,x_m)=g(x_2,\ldots,x_m,x_1)$ then
$g\circ [h,\ldots,h]$ also satisfies the same identity. The claim follows.
\end{enumerate}

\end{proof}


\section{Proof of Corollary~\ref{cor:gstar}}\label{sec:core:gstar:proof}
\paragraph{Part (1)}
Let ${\tt TEST}(\Gamma,\Pi)\in\{0,1\}$ be the predicate that equals $1$ if $\Pi$ is a partition of $\Gamma$,
and $0$ otherwise. Note that ${\tt TEST}(\Gamma,\Pi)\ge {\tt TEST}(\Gamma,\Pi')$ if $\Pi\preceq \Pi'$.
Clearly, a maximal partition $\Pi$ of $\Gamma$ can be found by a greedy
search that starts with $\Pi=\{\{a\}\:|\:a\in D\}$ and then makes at most $|D|(|D|-1)/2$ evaluations of ${\tt TEST}(\Gamma,\Pi)$ for partitions $\Pi\ne\{\{a\}\:|\:a\in D\}$.
To get the desired algorithm, we simply replace each evaluation of ${\tt TEST}(\Gamma,\Pi)$ in this greedy search
with a call to the algorithm in Theorem~\ref{th:Pi-testing} 
(where output~({\em{a}}) is treated as ``$1$'', and output~({\em{b}}) as ``$0$'').

\paragraph{Part (2)} Assume that we have a conditional maximal partition $\Pi$ of $\Gamma$.
Let us run the algorithm from Theorem~\ref{th:find-core:general} with $\sigma=\Pi$ and $\calB_\sigma=\BcorePi$.
Suppose that it returns a fractional polymorphism $\omega\in \fPol\Gamma$ with $\supp(\omega)\subseteq\calOplus[\BcorePi]$.
If it satisfies $\supp(\omega)\subseteq\OcorePi$ then we pick some $g\in \supp(\omega)$
and return set $B=g(D)$ together with vector $\omega$. 
In any other case we return set $B=\varnothing$, meaning that $\Gamma$
is not solvable. Next, we show the correctness of this procedure.
From now on we assume that $\Gamma$ is solvable (otherwise any $B$ is a conditional core of $\Gamma$
and so the claim is trivial). The assumption means $\Pi$ that maximal partition of $\Gamma$ and $\Ocore\Gamma=\OcorePi\cap \fPolplus\Gamma$
(by Lemma~\ref{lemma:Pi}).

Calling the algorithm from Theorem~\ref{th:find-core:general} 
may have four possible outcomes listed below. Note, we return a non-empty set $B$ only in the first case.
\begin{itemize}
\item[(i)] {\em The algorithm gives vector $\omega\in\fPol\Gamma$ with $\supp(\omega)\subseteq\OcorePi$.}

Recall that in this case we pick $g\in \supp(\omega)$ and return set $B=g(D)$ together with vector~$\omega$.
We then have  $g\in \OcorePi\cap \fPolplus\Gamma=\Ocore\Gamma$ and thus $B$ is a core of $\Gamma$.
\end{itemize}
In the last three cases we will arrive at a contradiction, meaning that $\Gamma$ is not solvable.
\begin{itemize}
\item[(ii)] {\em The algorithm gives vector $\omega\in\fPol\Gamma$ s.t.\ $\supp(\omega)$ contains operation $g\in \calOplus[\BcorePi]-\OcorePi$.}

From definitions,  $|g(D)|<|\Pi|$. However, we have $\Ocore\Gamma\subseteq \OcorePi$ and
thus $\coresize(\Gamma)=|\Pi|$ - a contradiction.
\item[(iii)] {\em The algorithm asserts that there is no $\omega\in\fPol\Gamma$ with $\supp(\omega)\subseteq\OcorePi$.}

By Lemma~\ref{lemma:core-basics}(a,b), there exists $\omega\in\fPol\Gamma$ with $\supp(\omega)\subseteq \Ocore\Gamma\subseteq\OcorePi$ -
 a contradiction.
\item[(iv)] {\em The algorithm asserts that either $\Gamma$ is not solvable or $\BcorePi\cap \Bcore\Gamma=\varnothing$.}

Since $\Ocore\Gamma\subseteq\OcorePi$, any $B\in\Bcore\Gamma$ satisfies $B\in\BcorePi$.
Thus, $\BcorePi\cap \Bcore\Gamma\ne\varnothing$ - a contradiction.
\end{itemize}

\section*{Acknowledgements}
The author is supported by the European Research Council under the European Unions Seventh Framework Programme (FP7/2007-2013)/ERC grant agreement no 616160.

\bibliographystyle{plain}
\bibliography{VCSP}

\newcommand{\noopsort}[1]{}
\begin{thebibliography}{10}

\bibitem{CSPsurvey:17}
L.~Barto, A.~Krokhin, and R.~Willard.
\newblock Polymorphisms, and how to use them.
\newblock In A.~Krokhin and S.~{\noopsort{ZZ}\v{Z}}ivn\'y, editors, {\em The
  Constraint Satisfaction Problem: Complexity and Approximability}. Dagstuhl
  Follow-Ups series, Volume 7, 2017.

\bibitem{barto11:lics}
Libor Barto.
\newblock The dichotomy for conservative constraint satisfaction problems
  revisited.
\newblock In {\em Proceedings of the 26th {I}{E}{E}{E} Symposium on Logic in
  Computer Science (LICS'11)}, pages 301--310. IEEE Computer Society, 2011.

\bibitem{barto09:siam}
Libor Barto, Marcin Kozik, and Todd Niven.
\newblock The {C}{S}{P} dichotomy holds for digraphs with no sources and no
  sinks (a positive answer to a conjecture of {B}ang-{J}ensen and {H}ell).
\newblock {\em {SIAM} Journal on Computing}, 38(5):1782--1802, 2009.

\bibitem{bjoerklund:09}
A.~Bj\"{o}rklund, T.~Husfeldt, and M.~Koivisto.
\newblock Set partitioning via inclusion–exclusion.
\newblock {\em SIAM J. Computing}, 39(2):546--563, 2009.

\bibitem{bulatov06:3-elementjacm}
Andrei Bulatov.
\newblock A dichotomy theorem for constraint satisfaction problems on a
  3-element set.
\newblock {\em Journal of the ACM}, 53(1):66--120, 2006.

\bibitem{Bulatov:FOCS17}
Andrei Bulatov.
\newblock A dichotomy theorem for nonuniform {CSP}s.
\newblock In {\em Proceedings of the 58th {A}nnual {I}{E}{E}{E} {S}ymposium on
  {F}oundations of {C}omputer {S}cience ({F}{O}{C}{S}'17)}. IEEE Computer
  Society, 2017.

\bibitem{bulatov05:classifying}
Andrei Bulatov, Andrei Krokhin, and Peter Jeavons.
\newblock Classifying the {C}omplexity of {C}onstraints using {F}inite
  {A}lgebras.
\newblock {\em {SIAM} Journal on Computing}, 34(3):720--742, 2005.

\bibitem{Bulatov11:conservative}
Andrei~A. Bulatov.
\newblock Complexity of conservative constraint satisfaction problems.
\newblock {\em ACM Transactions on Computational Logic}, 12(4), 2011.
\newblock Article 24.

\bibitem{SETH2}
Chris Calabro, Russell Impagliazzo, and Ramamohan Paturi.
\newblock The complexity of satisfiability of small depth circuits.
\newblock In {\em Proceedings of the 4th International Workshop on
  Parameterized and Exact Computation (IWPEC)}, volume 5917 of {\em LNCS},
  pages 75--85, 2009.

\bibitem{ChenLarose:17}
Hubie Chen and Benoit Larose.
\newblock Asking the metaquestions in constraint tractability.
\newblock {\em ACM Transactions on Computation Theory (TOCT)}, 9(3), October
  2017.

\bibitem{feder98:monotone}
Tom\'as Feder and Moshe~Y. Vardi.
\newblock The {C}omputational {S}tructure of {M}onotone {M}onadic {S{N}{P}} and
  {C}onstraint {S}atisfaction: {A} {S}tudy through {D}atalog and {G}roup
  {T}heory.
\newblock {\em {SIAM} Journal on Computing}, 28(1):57--104, 1998.

\bibitem{Fomin:ICALP15}
Fedor~V. Fomin, Alexander Golovnev, Alexander~S. Kulikov, and Ivan Mihajlin.
\newblock Lower bounds for the graph homomorphism problem.
\newblock In {\em Proceedings of the 36th International Colloquium on Automata,
  Languages and Programming (ICALP)}, volume 9134 of {\em LNCS}. Springer,
  2015.

\bibitem{GLS88:book}
Martin Gr\"otschel, L\'aszl\'o Lov\'asz, and Alexander Schrijver.
\newblock {\em Geometric algorithms and combinatorial optimization}.
\newblock Springer, 1988.

\bibitem{hell:92}
P.~Hell and J.~Ne\v{s}et\v{r}il.
\newblock The core of a graph.
\newblock {\em Discrete Mathematics}, 109(1-3):117--126, 1992.

\bibitem{SETH0}
Russell Impagliazzo and Ramamohan Paturi.
\newblock On the complexity of k-{SAT}.
\newblock {\em Journal of Computer and System Sciences}, 62(2):367--375, 2001.

\bibitem{SETH1}
Russell Impagliazzo, Ramamohan Paturi, and Francis Zane.
\newblock Which problems have strongly exponential complexity?
\newblock {\em Journal of Computer and System Sciences}, 63(4):512--530, 2001.

\bibitem{jonsson:jcss17}
P.~Jonsson, V.~Lagerkvist, G.~Nordh, , and B.~Zanuttini.
\newblock Strong partial clones and the time complexity of {SAT} problems.
\newblock {\em Journal of Computer and System Sciences}, 84:52--78, 2017.

\bibitem{jonsson:mfcs17}
Peter Jonsson, Victor Lagerkvist, and Biman Roy.
\newblock Time complexity of constraint satisfaction via universal algebra.
\newblock In {\em Proceedings of the 42nd International Symposium on
  Mathematical Foundations of Computer Science (MFCS)}, 2017.

\bibitem{Siggers:2}
Keith Kearnes, Petar Markovi\'c, and Ralph McKenzie.
\newblock Optimal strong {M}al'cev conditions for omitting type 1 in locally
  finite varieties.
\newblock {\em Algebra Universalis}, 72(1):91--100, 2014.

\bibitem{KKZ:SICOMP17}
Vladimir Kolmogorov, Andrei Krokhin, and Michal Rol\'inek.
\newblock The complexity of general-valued {CSP}s.
\newblock {\em SIAM Journal on Computing (SICOMP)}, 46(3):1087--1110, 2017.

\bibitem{Kolmogorov2015:hybrid}
Vladimir Kolmogorov, Michal Rol{\'i}nek, and Rustem Takhanov.
\newblock {\em Effectiveness of Structural Restrictions for Hybrid CSPs}, pages
  566--577.
\newblock Springer Berlin Heidelberg, Berlin, Heidelberg, 2015.

\bibitem{kolmogorov15:power}
Vladimir Kolmogorov, Johan Thapper, and Stanislav {\noopsort{ZZ}\v{Z}}ivn\'y.
\newblock The power of linear programming for general-valued {CSP}s.
\newblock {\em SIAM Journal on Computing}, 44(1):1–--36, 2015.

\bibitem{Kozik15:algebraic}
Marcin Kozik and Joanna Ochremiak.
\newblock Algebraic properties of valued constraint satisfaction problem.
\newblock arXiv1403.0476, 2015. Extended abstract published in {ICALP'15}.

\bibitem{Lagerkvist:arXiv18}
Victor Lagerkvist and Magnus Wahlstr\"{o}m.
\newblock Which {NP}-hard {SAT} and {CSP} problems admit exponentially improved
  algorithms?
\newblock arXiv1801.09488, 2018.

\bibitem{lawler:76}
E.L. Lawler.
\newblock A note on the complexity of the chromatic number problem.
\newblock {\em Inf. Process. Lett.}, 5(3):66--67, 1976.

\bibitem{lokshtanov11:survey}
Daniel Lokshtanov, D\'{a}niel Marx, and Saket Saurabh.
\newblock Lower bounds based on the exponential time hypothesis.
\newblock {\em Bull. EATCS}, 105:41--72, 2011.

\bibitem{razgon:05}
Igor Razgon.
\newblock Complexity analysis of heuristic {CSP} search algorithms.
\newblock In {\em International Workshop on Constraint Solving and Constraint
  Logic Programming (CSCLP)}, volume 3978 of {\em LNCS}, pages 88--99, 2005.

\bibitem{Schaefer78:complexity}
Thomas~J. Schaefer.
\newblock The {C}omplexity of {S}atisfiability {P}roblems.
\newblock In {\em Proceedings of the 10th {A}nnual {A}{C}{M} {S}ymposium on
  {T}heory of {C}omputing ({S}{T}{O}{C}'78)}, pages 216--226. ACM, 1978.

\bibitem{Siggers:1}
Mark~H. Siggers.
\newblock A strong {M}al'cev condition for locally finite varieties omitting
  the unary type.
\newblock {\em Algebra Universalis}, 64(1-2):15--20, 2010.

\bibitem{tz16:jacm}
Johan Thapper and Stanislav {\noopsort{ZZ}\v{Z}}ivn\'y.
\newblock The complexity of finite-valued {C}{S}{P}s.
\newblock {\em Journal of the ACM (JACM)}, 63(4), 2016.

\bibitem{traxler:08}
Patrick Traxler.
\newblock The time complexity of constraint satisfaction.
\newblock In {\em Proceedings of the 3rd International Workshop on
  Parameterized and Exact Computation (IWPEC)}, pages 190--201, 2008.

\bibitem{Zhuk:FOCS17}
Dmitriy Zhuk.
\newblock A proof of {CSP} dichotomy conjecture.
\newblock In {\em Proceedings of the 58th {A}nnual {I}{E}{E}{E} {S}ymposium on
  {F}oundations of {C}omputer {S}cience ({F}{O}{C}{S}'17)}. IEEE Computer
  Society, 2017.

\end{thebibliography}

\end{document}